\documentclass[12pt]{amsart}
\usepackage[]{color}
\usepackage{setspace}
\usepackage{wrapfig}
\usepackage{lineno}
\usepackage{graphicx}
\usepackage{amsmath}
\usepackage{amsfonts}
\usepackage{natbib}
\usepackage{hyperref}
\usepackage{subcaption}
\hypersetup{
    colorlinks=true,
    linkcolor=blue,
    filecolor=magenta,
    urlcolor=cyan,
    citecolor=blue,
}
\usepackage{enumitem}
\usepackage{url}
    \usepackage{amsmath,amsxtra,amssymb,latexsym,epsfig,amscd,amsthm,fancybox,epsfig}
\usepackage[mathscr]{eucal}
\usepackage{graphicx}
\usepackage{multicol,xcolor}
\usepackage{cases}
\usepackage{color}
\usepackage{hyperref}
\usepackage{bigints}

\DeclareFontFamily{U}{matha}{\hyphenchar\font45}
\DeclareFontShape{U}{matha}{m}{n}{
      <5> <6> <7> <8> <9> <10> gen * matha
      <10.95> matha10 <12> <14.4> <17.28> <20.74> <24.88> matha12
      }{}
\DeclareSymbolFont{matha}{U}{matha}{m}{n}

\DeclareMathSymbol{\Lt}{3}{matha}{"CE}
\DeclareMathSymbol{\Gt}{3}{matha}{"CF}

\newtheorem{thm}{Theorem}[section]
\newtheorem {asp}{Assumption}[section]

\newtheorem{rmk}{Remark}[section]

\newtheorem{cor}{Corollary}[section]

\newtheorem{prop}{Proposition}[section]
\theoremstyle{definition}

\theoremstyle{remark}

\DeclareMathOperator{\NN}{\mathbf N}

\newcommand{\E}{\mathbb{E}}

\newcommand{\BX}{\mathbf{X}}
\newcommand{\bx}{\mathbf{x}}

\newcommand{\N}{\mathbb{N}}

\newcommand{\PP}{\mathbb{P}}

\newcommand{\K}{\mathcal{K}}

\newcommand{\R}{\mathbb{R}}

\newcommand{\bed}{\begin{displaymath}}
\newcommand{\eed}{\end{displaymath}}
\newcommand{\bea}{\bed\begin{array}{rl}}
\newcommand{\eea}{\end{array}\eed}

\newcommand{\barray}{\begin{array}{ll}}
\newcommand{\earray}{\end{array}}
\newcommand{\diag}{{\rm diag}}

\def\bar{\overline}
\def\hat{\widehat}
\def\a.s{\text{\;a.s.\;}}
\def\supp{\text{supp\,}}

\newcommand{\Se}{\mathcal{S}}

\usepackage{natbib}
\begin{document}
\bibliographystyle{agsm}

\title{Population size in stochastic multi-patch ecological models}

\author[A. Hening]{Alexandru Hening }
\address{Department of Mathematics\\
Texas A\&M University\\
Mailstop 3368\\
College Station, TX 77843-3368\\
United States
}
\email{ahening@tamu.edu}

\author[S. Sabharwal]{Siddharth Sabharwal}
\address{Department of Mathematics\\
Texas A\&M University\\
Mailstop 3368\\
College Station, TX 77843-3368\\
United States
}
\email{siddhutifr93@tamu.edu }

\begin{abstract}
We look at the interaction of dispersal and environmental stochasticity in $n$-patch models. We are able to prove persistence and extinction results even in the setting when the dispersal rates are stochastic. As applications we look at Beverton-Holt and Hassell functional responses. We find explicit approximations for the total population size at stationarity when we look at slow and fast dispersal. In particular, we show that if dispersal is small then in the Beverton-Holt setting, if the carrying capacity is random, then environmental fluctuations are always detrimental and decrease the total population size. Instead, in the Hassell setting, if the inverse of the carrying capacity is made random, then environmental fluctuations always increase the population size. Fast dispersal can save populations from extinction and therefore increase the total population size. We also analyze a different type of environmental fluctuation which comes from switching environmental states according to a Markov chain and find explicit approximations when the switching is either fast or slow - in examples we are able to show that slow switching leads to a higher population size than fast switching.

Using and modifying some approximation results due to Cuello, we find expressions for the total population size in the $n=2$ patch setting when the growth rates, carrying capacities, and dispersal rates are influenced by random fluctuations. We find that there is a complicated interaction between the various terms and that the covariances between the various random parameters (growth rate, carrying capacity, dispersal rate) play a key role in whether we get an increase or a decrease in the total population size. Environmental fluctuations turn to sometimes be beneficial -- this shows that not only dispersal, but also environmental stochasticity can lead to an increase in population size.

\end{abstract}


\maketitle

\section{Introduction} \label{s:intro}

The effects of environmental stochasticity cannot be neglected if one wants to have realistic mathematical models of population dynamics. As a result, there has been a sustained effort  \citep{B18, BS19, HNC21, FS24} to generalize Chesson's Modern Coexistence Theory (MCT) to complex ecosystems with arbitrary levels of noise. 

It is well known that human activities can lead to the fragmentation of habitats into isolate patches. As a conservation problem, it is important to know whether higher dispersal i.e. increased connectivity would lead to higher population sizes, especially if one looks at endangered species. Empirical evidence suggests that there is a complex relationship between dispersal and the total population size at equilibrium. In certain cases higher dispersal increases total population size \citep{ives2004synergistic,zhang2017carrying} while in others it decreases it \citep{aastrom2013negative}. There are also examples where at first the total population size increases due to dispersal and then decreases \citep{vortkamp2022dispersal, grumbach2023effect}.

Quite often, the environmental factors influencing a population that is spread spatially can differ according to the spatial location. In many situations the region can be analyzed as a network of connected patches, each having a subpopulation of the species. The subpopulations can then move or \textit{disperse} between the various patches, which have different characteristics (growth/death rates etc). Biologists have been interested in the effect dispersal has on population growth and persistence \citep{H83, C85, C00, GH02, S04, RHB05, schreiber2010interactive, CCL12, DR12}.

This paper is in part inspired by the deterministic work from \cite{grumbach2023effect} where the authors looked at the effect of dispersal on total population size in some simple deterministic two-patch models. Our goal is to look at the combined effects of dispersal and stochasticity on  total population size. For this we focus on $n$-patch dispersal models and consider settings where multiple parameters are affected by environmental fluctuations, including the possibility of the dispersal rate itself being stochastic. There are fewer rigorous mathematical results  in the setting where the dispersal rates are stochastic. One recent exception is \cite{schreiber2023partitioning} where the author is able to prove approximation results for linear stochastic matrix population models with stochastic dispersal.  

The general theory of stochastic persistence that is now available gives conditions under which interacting species persist and converge to an invariant probability measure, or stationary distribution. However, the current theory does not usually say anything about the `persistence measure', or in particular, about the average population sizes of the persistent species at stationarity. In this paper, which is a follow up to \cite{HS25a} where we studied similar questions in the absence of dispersal, we look for analytic expressions for the average population size. In order to handle this problem we look at some recent approximation results developed by \cite{cuello2019persistence}. Cuello shows that under certain natural assumptions, if the noise is small, one can do a Taylor-type approximation around the deterministic stable fixed point (which describes the population sizes in the absence of stochasticity) and find explicitly the first and second order terms. We will use this result to see how some important models behave under the influence of small environmental stochasticity. 

There are related small-noise results which have been used in other settings. In \cite{barbier2017cavity} the authors use a disordered systems approach from statistical physics in order to look at small perturbations of the carrying capacities and how these change the abundances of the species. Some of the questions we ask in the current paper have already been asked and in part answered \citep{mallmin2024chaotic} for other types of models. These models are usually deterministic and the fluctuations of the population are endogenous, rather than driven by the environment like in our setting. Future ecological work should take advantage of both our approach and those cited above.

The current paper builds on important previous results such as \cite{kirkland2006evolution, BS09, RS14, cuello2019persistence, HNC21, grumbach2023effect} and the recent paper \cite{HS25a}. The idea of looking at how stochasticity and dispersal influence populations is not new, but it is new if one looks specifically at the total population size at stationarity. In the current paper we are putting together all the pieces, and sometimes provide some generalizations to previous mathematical results. However, the main contribution of this paper is to see how everything fits together in order to get new biological insights, from the deterministic results of \cite{kirkland2006evolution, grumbach2023effect}, to the stochastic results from \cite{BS09, RS14, HNC21}, and the small-noise results of \cite{cuello2019persistence}. More specifically:
\begin{itemize}
    \item We generalize the results from \cite{BS09} so that they can also work for random dispersal and show that the assumptions hold for both Beverton-Holt and Hassell functional responses. 
    \item We use results from \cite{BS09} in order to show how slow and fast dispersal interacts with randomness. A result from \cite{HNC21} allows us to see that, not only the Lyapunov exponents and invasion rates, but also the expectations of the populations at stationarity depend continuously on the dispersal rate. This implies that for small dispersal and arbitrary noise, in the Beverton-Holt $n$-patch model noise decreases the total population size. For large dispersal (close to $1$) we show that dispersal can save the population from extinction and therefore can increase the total population size. 
    \item We show how everything is well-defined in the stochastic version of the two-patch model from \cite{grumbach2023effect} and that coupling two sink patches through dispersal can save a population from extinction.
    \item In Section 3 we present some results when the randomness comes, not from i.i.d environments, but from a discrete-time Markov chain which switches between finitely many states $\{1,\dots,n_0\}$. In particular, we look at how slow switching, where the process spends a long time in a state before jumping to another state, and fast switching, where the process switches with probability $1$ at every time step, changes the total population size at stationarity. This yields explicit approximations for the total population size for $n$-patch models. 
    \item In Section 4 we return to i.i.d. noise and consider small-noise approximations. We use and expand the results from \cite{cuello2019persistence} -- we are able to generalize \cite{cuello2019persistence} to non-zero mean noise and to the non-Kolmogorov dispersal form. The applications involve using, adapting and generalizing results from \cite{cuello2019persistence, kirkland2006evolution, BS09}. We get explicit expressions for the approximations and therefore get important biological insights.
    \item This paper is the natural follow-up of \cite{HS25a} where we looked at how noise influences the population size at stationarity when there is no dispersal.
\end{itemize}

The paper is structured as follows. We introduce the models and present our persistence results for $n$-patch models with arbitrary noise in Section \ref{s:dispersal}. In Section \ref{s:small_noise} we use small-noise expansions of deterministic difference equations in order to approximate the invariant distribution of persistent systems. This allows us to quantify how environmental stochasticity changes the total population size.  We also get quantitative approximations for small noise and small, possibly random, dispersal. Numerical experiments, which go beyond small noise, are done in Section \ref{s:sim}.
We end our paper with a discussion of our results and future challenges in Section \ref{s:discussion}.

\section{The interaction of dispersal and stochasticity}\label{s:dispersal}

Populations are usually spread out heterogeneously throughout space. It is therefore important to introduce some type of spatial dynamics in any realistic ecological modeling framework. One way of doing this is by including dispersal of the population between various habitat patches. The various patches will have different  conditions, since both abiotic (sunlight, nutrients) and biotic (predation, competition, cooperation) factors can change from patch to patch. Individual organisms can then disperse between the various patches in order to have a better survival chance.

As the environment changes it can affect not only the birth, death, and interaction rates of various species but also their dispersal rates. The joint effects of fluctuations in demography and dispersal on population growth have been studied extensively \citep{levin1984dispersal,roy2005temporal, schreiber2010interactive, evans2013stochastic, HNY16}. The recent study \cite{schreiber2023partitioning} provides a powerful way of analyzing the impacts of spatio-temporal variation in dispersal and demographic rates on the growth rate of a population spread throughout $n$ patches. Our focus will be to analyze how the fluctuations of the dispersal rates and the demographic rates impact the total population size at stationarity. This work can be seen as a continuation of the deterministic work by \cite{grumbach2023effect} which looked at how dispersal changes the total population size.

\subsection{Deterministic $n$ patch model} Suppose we have a population spread throughout $n$ patches. The population size in patch $i$ at time $t$ will be denoted by $N_i(t)$. One way of modeling deterministically the dispersal dynamics  \citep{kirkland2006evolution, schreiber2010interactive, BS09} is given by
\begin{equation}\label{e:dispersal_det}
N_i(t+1) = \left(1-\sum_{k\neq i}d_{ki}\right)f_i(N_i(t))N_i(t) + \sum_{j\neq i} d_{ij} f_j(N_j(t))N_j(t), ~i=1,\dots,n.
\end{equation}
Here $d_{ij}\geq 0$ denotes the fraction of the population from patch $j$ that disperses to patch $i$ and
\[
d_{ii}=\left(1-\sum_{k\neq i}d_{ki}\right)
\]
is the fraction of the population which remains in patch $i$. Here $f(N_i)$ are the functional responses which give the per-capita growth rates in patch $i$ as functions of the population abundance in patch $i$. 
Define $D:=(d_{ij})$, the diagonal matrix $\Lambda(N):=\diag (f_1( N_1),\dots,f_n( N_n))$ and the matrix $A(\NN):=D\Lambda(\NN)$. We can then write \eqref{e:dispersal_det} as
\begin{equation}\label{e:dispersal_det2}
\NN(t+1) = D\Lambda(\NN(t))\NN(t)= A(\NN(t))\NN(t) := F(\NN(t)).
\end{equation}

We assume throughout the paper that $D$ is primitive i.e. $D^k>0$ (all the entries are strictly positive) for some $k\in \N$. This means that after $k$ time steps, individuals can move between any two patches.
For a matrix $A$ let $\lambda(A)$ denote its spectral radius. A result from \cite{kirkland2006evolution} shows that under natural assumptions the long-term behavior of \eqref{e:dispersal_det} is fully determined by $\lambda(D \Lambda(0))$.
\begin{thm} (\cite{kirkland2006evolution})\label{t:dispersal_det}
Assume that
\begin{itemize}
    \item [1)] The $f_i$'s are strictly positive continuous decreasing functions.
    \item [2)] $\lim_{x\to\infty}f_i(x)<1$.
    \item [3)] The functions $x\mapsto x f_i(x)$ are increasing.
    \item [4)] The matrix $D$ is primitive.
\end{itemize}
The following classification of the dynamics holds:
\begin{itemize}
    \item [i)] If $\lambda(D \Lambda(0))>0$ there is $\bar N>0$ such that $N(t)\to \bar N$ as $t\to\infty$ for all $N(0)\in \R_{++}^n$.
    \item [ii)] If $\lambda(D \Lambda(0))\leq 0$ then $N(t)\to 0$ as $t\to\infty$ for all $N(0)\in \R_{+}^n$.
\end{itemize}
\end{thm}
\begin{rmk}\label{r:func}
The assumptions say that the per-capita growth rates decrease with the population density, that the population itself has to decrease if the population density is high enough and that a higher population density at time $t$ always leads to a higher density at time $t+1$. 
We note that the Beverton-Holt $f_i( x) = \frac{r_i  }{1+\kappa_i x}$, Ivlev $f_i( x) = r_i(1-\exp(-b_i x))$ and Hassel $f_i(x)=\frac{\alpha_i}{(1+\K_ix)^c}$ (provided $c<1$) functional responses satisfy the assumptions. In the Beverton-Holt model $r_i$ is the intrinsic growth rate and the intraspecific competition strength is quantified by $\kappa_i$ while the carrying capacity $K_i$ is given by $K_i=\frac{r_i-1}{\kappa_i}.$ In the Hassell model $\alpha_i$ represents the growth rate for small populations, $\K_i$ is the competition rate and $c$ determines the type of density dependence: $c=1$ is scramble competition, $c<1$ means undercompensation (weak density dependence), while $c>1$ implies overcompensation (strong density dependence). 
\end{rmk}
\subsection{Stochastic $n$-patch model}
The model \eqref{e:dispersal_det} becomes naturally stochastic if we let the dispersal rates and the functional responses depend on a stochastic environment $(\xi(t))_{t\geq 0}$ which forms an i.i.d sequence in $\R^w$ for some $w\in \N$. As a result
\begin{equation}\label{e:dispersal_stoch}
N_i(t+1) =  \sum_{j=1}^n d_{ij}(\xi(t)) f_j(N_j(t),\xi(t))N_j(t), ~i=1,\dots,n
\end{equation}
or, in matrix notation,
\begin{equation}\label{e:dispersal_stoch2}
\NN(t+1) :=G((\NN(t),\xi(t)))= D(\xi(t))\Lambda(\NN(t),\xi(t))\NN(t)= A(\NN(t),\xi(t))\NN(t).
\end{equation}
Following \cite{BS09}, the main way of analyzing this dynamics is by linearizing the system around zero and looking at the the linearized version
\[
\tilde \NN(t+1) = A(0,\xi(t))\tilde \NN(t)= \left(\prod_{s=0}^t A(0,\xi(s))\right) \tilde \NN(0).
\]
\begin{rmk}\label{r:notation}
We will slightly abuse notation as we sometimes will write
\[
 D(\xi(t))
\]
while at other times
\[
D(t)
\]
At certain points we will even write $\xi(t)=(r_i(t), \kappa_i(t), D(t))$ - see Assumption \ref{a:BH_dd}. Nevertheless, we think this abuse of notation actually makes the paper more readable as it is always clear from the context as to what is random and what is not random. 
\end{rmk}

\begin{asp}\label{asp:lyapunov}
With probability one the matrix $D$ is primitive, $\Lambda$ has a strictly positive diagonal and 
\[
\E \left| \ln \|A(0,\xi(0)\|\right|<\infty.
\]
\end{asp}
Under Assumption \ref{asp:lyapunov} one can show \citep{BS09} that the Lyapunov exponent exists. This means that there exists $M\in\R$ such that with probability one
\[
\lim_{t\to \infty }\frac{1}{t}\ln\left(\tilde N_1(t)+\dots +\tilde N_n(t)\right)=M,
\]
for any $\tilde\NN(0)\neq 0.$
One can then prove \citep{BS09,RS14, BSt19, HNC21} that $M$, the metapopulation growth rate, determines the long-term behavior of the nonlinear system \eqref{e:dispersal_stoch2}. In particular, we will make use of the following result from \cite{BS09}.
\begin{asp}\label{a:BS09}
Suppose the following hold.
\begin{itemize}
    \item There exists a function $V:\R_+^k\to \R$ and random variables $\alpha, \beta:\R^w\to [0,\infty)$ such that
    \begin{enumerate}
        \item $V(A(\xi,x)x) \leq \alpha(\xi)V(x) + \beta(\xi)$;
        \item $\E (\ln(\alpha))<0$;
        \item $\E(\ln^+\beta)<\infty$.
    \end{enumerate}
    \item The map $(\xi,x)\to A(\xi,x)$ is Borel, $G$ is twice continuously differentiable for all $\xi \in \R^w, x\in\R_+^k$, and
    \[
    \E\left(\sup_{\|x\|\leq 1}\ln^+(\|G((x,\xi(1)))\| + \|DG((x,\xi(1)))\| + \|D^2G((x,\xi(1)))\|)\right)<\infty.
    \]
    \item The matrix entries $A_{ij}(x,\xi)$ satisfy that $\frac{\partial A_{ij}}{\partial x_l}(x,\xi)\leq 0$ for all $\xi, x$ and $l$ and for every $i$ there are $j, l$ such that the inequality becomes $\frac{A_{ij}}{\partial x_l}(x,\xi)< 0$ for all $\xi, x$.
    \item All entries of the derivative $DG(x,\xi)$ of $G(x,\xi)$ are non-negative for all $\xi$ and $x$.
\end{itemize}
\end{asp}
Let 
\[
\Pi(T):= \frac{1}{T}\sum_{t=1}^T \delta_{\NN(t)}
\]
be the normalized occupation measures of the process $\NN(t)$. The following result follows from Theorem 1, Theorem 3, and Proposition 1 from \cite{BS09}.
\begin{thm}\label{t:BS09}
Assume Assumptions \ref{asp:lyapunov} and \ref{a:BS09} hold. 
\begin{enumerate}
    \item If $M>0$ then there exists a positive random vector $\hat \NN$ with distribution $\mu$ such that $\NN(t)\to \hat \NN$ whenever $\NN(0)>0.$ Furthermore, one has that with probability one $\Pi(t)\to \mu$ as $t\to\infty$.  

    \item If $M<0$ then for any $\NN(0)$ we have that with probability one $\NN(t)\to 0$ as $t\to\infty$.
\end{enumerate}
\end{thm}

\begin{rmk}
Because of the primitivity assumption, it is impossible to have one patch go extinct without the other patches going extinct. Primitivity ensures that all patches go extinct or all survive. This makes proving persistence in this setting easier than for a truly $n$-dimensional ($n$ species) model.
\end{rmk}

\begin{rmk}
Usually it is almost always impossible to find $M$ explicitly but there are ways to estimate or approximate it \citep{BS09}. In specific cases we can even find $M$ explicitly - see Section \ref{s:SF}.      
\end{rmk}

\subsection{Stochastic $n$-patch Beverton Holt model}
A widely used functional response is the Beverton-Holt one.  This model has been used extensively and can be seen as a discrete time analogue of the logistic equation. Beverton and Holt used it in 1957 to analyze exploited fish populations \citep{beverton2012dynamics} but this functional response has been shown to work well in general with both contest and scramble competition models \citep{brannstrom2005role}.

Using a result from \cite{BS09} allows us to fully classify the dynamics in the Beverton-Holt $n$-patch setting.

Suppose the dynamics is given by
\begin{equation} \label{e:BV_disper_stoc}
\NN(t+1) =F(\NN(t),\xi(t)):=\left(D(t)\diag\left(  \frac{r_1(t)}{1+\kappa_1(t)N_1(t)},\dots,\frac{r_n(t)}{1+\kappa_n(t)N_n(t)}\right)\right)\NN(t)
\end{equation}
where $D(t):=(d_{ij}(t)), r_i(t)$ and $\kappa_i(t)$ can be random.

\begin{asp}\label{a:BH_dd}
The random environmental fluctuations affect the model parameters so that $$\xi(t)=(r_i(t), \kappa_i(t), d_{ij}(t))_{i, j=1\dots n}.$$
We will assume the following:
\begin{enumerate}
\item $(\xi(t))_{t\geq 0}$ forms an iid sequence in $(0,\infty)^n \times (0,\infty)^n \times [0,1]^{n^2}\subset \R^\ell$ where $\ell=2n+n^2$.
    \item $\E\left(\ln^+\left(\max_j \frac{r_j}{\kappa_j}+ \max_j r_j + 2 \max_j (r_j\kappa_j)\right)\right)<\infty$.
    \item With probability one: $r_i(1)>0, \kappa_i(1)> 0, 0<d_{ij}(1)<1$ and $\sum_{j=1}^nd_{ij}(1)=1$ for $i,j=1,\dots n$. 
\end{enumerate}
\end{asp}
\begin{rmk}
The first assumption on the noise is quite natural -- it can be relaxed to having an ergodic stationary sequence $(\xi(t))$ \citep{BS09}. Assumption (2) is required in order to make sure that the populations are stochastically bounded, and do not explode. The primitivity assumption (3) ensures that after the first step one can move between any two patches. These assumptions are satisfied in many ecological settings.
\end{rmk}

The next theorem gives us conditions under which the stochastic model \eqref{e:dispersal_stoch} exhibits persistence. Since our goal is to see how stochasticity influences total population size, we also need a deterministic model with which to compare the stochastic one with. The most natural choice seems to be the deterministic model where we let the coefficients be the averages of the random versions. It turns out that the deterministic model has a unique globally attracting fixed point. This shows that both models converge and therefore allow us to compare the average total population sizes. 

\begin{thm}\label{t:stoc_m}
Suppose the dynamics is given by \eqref{e:BV_disper_stoc} and Assumption \ref{a:BH_dd} holds. Then the Lyapunov exponent $M$ exists. If $M<0$ then with probability one $\NN(t)\to 0$ as $t\to \infty$ for any $\NN(0)\in\R_+^n$.  If $M>0$, then there exists a positive random vector $\NN(\infty)$ and $\NN(t)\to \NN(\infty)$ in distribution as $t\to\infty$. The total expected population size at stationarity will then be given by
 \[
\sum_{i=1}^n\E N_i(\infty).
 \]
 Moreover, if we look at the deterministic system
 \begin{equation}\label{e:dispersal_det_stoch}
\bar N_i(t+1) =  \sum_{j} \bar d_{ij} \bar f_j(\bar N_j(t))\bar N_j(t), ~i=1,\dots,n
\end{equation}
where $\bar f_i(x) = \frac{\bar r_i}{1+\bar \kappa_i x}$ and $\bar r_i = \E r_i(1), \bar K_i = \E K_i(1), \bar d_{ij} = \E d_{ij}(1)$, then this system has a unique globally attracting fixed point $(\bar N_1,\dots, \bar N_n)$ and the total change in expected population size due to environmental fluctuations is well-defined and is given by
\[
H= \sum_{i=1}^n\E N_i(\infty) -  \sum_{i=1}^n \bar N_i.
\]
 
\end{thm}
The proof of this result follows from a slight modification of \cite{BS09} and appears in Appendix \ref{a:BH_d}.
\begin{rmk}
There are multiple ways in which one can compare the stochastic dynamics with the deterministic one. Here we chose to compare to the deterministic dynamics that comes from replacing the stochastic growth rates and competition rates by their averages. One could instead replace the fitness functions $f_i(N(t))=\frac{r_i(t)}{1+\kappa_i(t)N(t)}$ by their averages $\hat f_i(x)= \E \left(\frac{r_i(1)}{1+\kappa_i(1)x}\right)$. For related comparisons between the stochastic and deterministic systems see \cite{hening2021coexistence, HS25a}.  
\end{rmk}

\subsection{Slow and fast dispersal in the Beverton-Holt $n$ patch model}\label{s:SF}
Suppose we simplify things by assuming that at each time step a fraction $\eta\in[0,1]$ of the population disperses and dispersed individuals are equally likely to pick any of the $n$ patches. We then get

\begin{equation} \label{e:BV_disper_stoc_unif}
\NN(t+1) =F_\eta(\NN(t),\xi(t)):=\left([(1-\eta)I_n + \eta U_n]\diag\left(  \frac{r_1(t)}{1+\kappa_1(t)N_1(t)},\dots,\frac{r_n(t)}{1+\kappa_n(t)N_n(t)}\right)\right)\NN(t)
\end{equation}
where $I_n$ is the identity $n\times n$ matrix and $U_n$ is an $n\times n$ matrix which has all its entries equal to $\frac{1}{n}.$
It can be shown \cite{BS09} that Assumption \ref{a:BS09} holds in this setting as long as 
$$\E\left(\ln^+\left(\max_j \frac{r_j}{\kappa_j}+ \max_j r_j + 2 \max_j (r_j\kappa_j)\right)\right)<\infty.$$ Moreover, by \cite{BS09} the Lyapunov exponent $M(\eta)$ exists. In in the small dispersal limit $\eta=0$ one can find \citep{BS09} the Lyapunov exponent explicitly, namely, 
\begin{equation}\label{e:M(0)}
    M(0)=\max_i\E \ln r_i(0).
\end{equation}

Moreover, it can be shown \citep{BS09} that $M(\eta)$ depends continuously on the dispersal rate so that the function 
\begin{equation}\label{e:cts}
\eta \mapsto M(\eta)
\end{equation}

is continuous. By Theorem 2.8 from \cite{HNC21} we get that the the models are robust, in the sense that, small perturbations of the models lead to similar behavior. In addition, from the robustness proof, one can see that the function 
\begin{equation}\label{e:robust}
\eta \mapsto \E N_i(\infty,\eta)
\end{equation}
is continuous. Here the $\eta$ in $N_i(\infty,\eta)$ quantifies the dependence of the distribution of $N_i(\infty)$ on the dispersal fraction $\eta$.

The total population size at stationarity becomes
\[
\sum_i\E N_i(\infty). 
\]
We note that if $\E\ln r_i(1) < 0$ then patch $i$ is a sink and goes extinct on its own so that $N_i(\infty)=0$. Nevertheless, according to \eqref{e:M(0)} even if one patch $i^*$ is a source i.e. $\E\ln r_{i^*}(1) < 0$, then the total population persists. This shows that for $\eta$ close enough to $0$ it is possible to get one single patch to `save' all the other patches from extinction. 
\begin{rmk}\label{r:BH}
    
If we write the Beverton-Holt dynamics in the equivalent form (see Remark \ref{r:func})
\[
N_{t+1}= \frac{m K(t)}{K(t) +(m-1)N(t)} N(t)
\]
and assume $K(t)$, the carrying capacity, is random and forms an iid sequence $(K_t)_{t\geq 1}$.  Let $\bar K = \E K_1$ and assume that $m>1$. The deterministic system 
\[
\bar N(t+1)= \frac{m \bar K}{\bar K +(m-1) \bar N(t)t} \bar N(t)
\]
converges as $t\to \infty$ to the fixed point given by the carrying capacity $\bar N(t)\to \bar K$. 
If $r(\delta_0)=\ln m >0$, which is the same assumption we needed for the deterministic model above i.e. $m>1$, using the results from \cite{HS25a} or the results from \cite{BS09}), one can show that the process $(N(t))$ has a unique stationary distribution $\mu$ on $\R_{++}$. By the proof of the Cushing-Henson conjecture \citep{HS05}, if only the carrying capacities are random, the noise is always detrimental and decreases the expected population size in the Beverton-Holt setting, that is 
\begin{equation}\label{e:ub} 
\E N(\infty)< \E K_1=\bar K.
\end{equation}
\end{rmk}
Note that by Remark \ref{r:BH} if only the carrying capacity fluctuates and the growth rates are constant in a patch $i$ then noise always \textit{decreases} the total population size in patch $i$. If $K_i(t)$ is the carrying capacity in patch $i$ and $r_i(1)=m_i>1$ then we can see that
\[
\sum_i \E K_i(1) > \sum_i \E N_i(\infty,\eta)
\]
for all $\eta$ close enough to $0$, where we sum over all the $i$ which are source patches i.e. satisfy $m_i>1$. This shows that in this setting for \textit{small dispersal and any type of noise, the total expected population size at stationarity is strictly lower than the total population size without noise.}

One can see that in the small dispersal limit, since the patches are not connected, we can treat each patch separately. As a result, it can be shown \citep{HS25a} using Birkhoff's ergodic theorem that , if $\mu_i$ is the invariant probability measure describing patch $i$ at stationarity, then
\[
r_i(\mu_i) =\int_{\R_+}\E[\ln f_i(x,\xi(1))]\,\mu_i(dx) = \E \ln\left(\frac{r_i(1)}{1+\kappa_i(1)N_i(\infty)}\right).
\]
In the high dispersal limit \citep{BS09} $\eta=1$ the Lyapunov exponent is
\[
M(1)=\E \ln \left(\frac{1}{n}\sum_{i=1}^nr_i(1)\right).
\]
In this case it is possible for all the patches to be sinks with $\E\ln r_i(1) < 0$ and  have persistence of the dispersal coupled system since Jensen's inequality implies that
\[
M(1)=\E \ln \left(\frac{1}{n}\sum_{i=1}^nr_i(1)\right) > \frac{1}{n} \sum_{i=1}^n \E \ln r_i(1)
\]
If $M(1)>0$ we therefore get that $N_i(t)\to N_i(\infty)>0$ for all dispersal rates $\eta$ that are close enough to $1$.
This shows that dispersal can \textit{save the population from extinction and increase the total population size}.

\subsection{Two-patch model}
Following \cite{grumbach2023effect} we are interested in the dynamics of a population that disperses between two patches. If we assume that the dispersal rate is identical in both directions and that the populations of the two patches reproduce with their own growth functions we get the difference equations

\begin{eqnarray}\label{e:det_disc}
\bar N_1(t+1) &=& (1-\delta)f_1(\bar N_1(t))\bar N_1(t) + \delta f_2(\bar N_2(t))\bar N_2(t)\\ \nonumber
\bar N_2(t+1) &=& (1-\delta)f_2(\bar N_2(t))\bar N_2(t) + \delta f_1(\bar N_1(t))\bar N_1(t).
\end{eqnarray}
We denote the total population size by $\bar N(t)=\bar N_1(t)+\bar N_2(t)$. We are mostly interested in the Beverton-Holt functional responses
\begin{equation}\label{e:BH}
f_i(\bar N_i) = \frac{r_i}{1+\kappa_i \bar N_i}
\end{equation}
where $r_i$ are the intrinsic growth rates and $\kappa_i=\frac{r_i-1}{K_i}$ are the competition strengths. Note that, just as in Remark \ref{r:func}, the $K_i$'s can be interpreted as the carrying capacities.

The following result from \cite{grumbach2023effect}, which is a corollary of Theorem \ref{t:dispersal_det}, tells us what happens long-term with the deterministic dynamics.

\begin{prop}\label{p:grumbach}
Assume we are in the Beverton-Holt setting and $1< r_2\leq r_1$. For each $\delta \in [0,1]$ system \eqref{e:det_disc} has a unique globally stable fixed point $(\bar N_1^\delta, \bar N_2^\delta)$ such that
\[
\lim_{t\to\infty}(\bar N_1(t), \bar N_2(t))=(\bar N_1^\delta, \bar N_2^\delta)
\]
for all $(N_1(0), N_2(0))\in \R_+^2\setminus\{(0,0)\}$.
\end{prop}
Next, assume there are environmental fluctuations which affect the dynamics. As a result \begin{eqnarray}\label{e:stoc_disc_matrix}
N_1(t+1) &=& (1-\delta(t))f_1(N_1(t), \xi(t))N_1(t) + \delta(t) f_2(N_2(t), \xi(t))N_2(t)\\ \nonumber
N_2(t+1) &=& (1-\delta(t))f_2(N_2(t), \xi(t))N_2(t)+ \delta(t) f_1(N_1(t), \xi(t))N_1(t).
\end{eqnarray}
Using matrix notation we set 
\[
D(t):= \begin{bmatrix}
     1-\delta(t)   &  \delta(t)     \\  
     \delta(t)       &    1-\delta(t)      
    \end{bmatrix} 
\]    
and 
\[
{\Lambda(N(t),\xi(t))}:= \begin{bmatrix}
    f_1( N_1(t),\xi(t))   &  0     \\  
     0       &      f_2( N_2(t),\xi(t))  
    \end{bmatrix} 
\]
where $\mathbf{N} := (N_1, N_2)^T$. 
We have 
\[
 A(\NN(t),\xi(t)):={D}(t){\Lambda(N(t),\xi(t))}.
\]
and
\[
\NN(t+1) = A(\NN(t),\xi(t)) \NN(t).
\]
As discussed above, one usually linearizes the system and gets that if
\[
\NN(T) =\prod_{k=0}^T A(0,\xi(k)) \NN(0).
\]
then there is $M\in \R$ (Theorem \ref{t:stoc_m}) such that 
\[
\PP\left(\lim_{T\to \infty}\frac{1}{T} \ln(N_1(T)+N_2(T)) = M\right) = 1.
\]

The following result is an immediate corollary of Theorem \ref{t:stoc_m}.

\begin{cor}
Suppose the dynamics is given by the system \eqref{e:stoc_disc_matrix} and Assumption \ref{a:BH_dd} holds. Then the Lyapunov exponent $M$ exists. If $M>0$ there is a positive vector $\NN_\infty$ such that for all nonzero initial population sizes $ \NN(0)$ one has
\[
 \NN (t) \to  \NN(\infty)
\]
in distribution. As a result, the total expected population size at stationarity is
\[
\E N_1(\infty) + \E N_2(\infty).
\]
\end{cor}

\begin{rmk}
By \cite{schreiber2010interactive} and \cite{metz1983advantages} one can see that if $\delta = 1/2$ in the Beverton-Holt setting one has
\[
M = \E \ln( 0.5 r_1+0.5   r_2).
\]
Using Jensen's inequality we further have
\begin{equation}\label{e:sink}
\ln\left(0.5 \E r_1 + 0.5 \E r_2 \right) > M>0.5 \E \ln r_1 +0.5\E \ln r_2.
\end{equation}
This shows that dispersal can facilitate persistence even when one of the two patches is a sink patch i.e. a patch where the population would go extinct without dispersal. For example, if $\E \ln r_1<0$ but $\E \ln r_1+ \E \ln r_2>0$ then patch $1$ is a sink patch yet the two patches coupled via dispersal exhibit persistence. Moreover, an even more interesting phenomenon can occur: it is actually possible to have $\E \ln r_1<0$ and $\E \ln r_2<0$ while $M = \E \ln( 0.5 r_1+0.5   r_2)>0$, showing that coupling two sink patches through dispersal can lead to persistence \citep{schreiber2010interactive, BS09}. 

Nevertheless, in general it is impossible to find analytical expressions for $M$ and one has to estimate $M$ (see \cite{BS09}). 
\end{rmk}

\subsection{Hassell functional response}
The Hassell model \citep{hassell1975density} captures important dynamical behaviors. The population undergoes exponential growth when it is small and density dependence to reduce the rate of growth when the the abundance is large. 
This model is especially useful \citep{hassell1975density,alstad1995managing}  because  it exhibits a wide range of density-dependent effects.

We chose the Hassell model because, in a previous paper that did not include dispersal \citep{HS25a}, we were able to show that environmental fluctuations can lead in this model to an increase in total population size. It is therefore natural to look at both the Beverton-Holt and Hassell functional responses in the setting with dispersal.
 
 Suppose
\begin{equation}\label{e:disper_hass}
    \begin{split}
\NN(t+1) &=F(\NN(t),\xi(t))= D(t)\Lambda(t) \NN(t)\\
&=\left(D(t)\diag\left(  \frac{\alpha_1(t)}{(1+K_1(t)N_1(t))^c},\dots,\frac{\alpha_n(t)}{(1+K_n(t)N_n(t))^c}\right)\right)\NN(t)
    \end{split}
\end{equation}
where $D(t):=(d_{ij}(t)), r_i(t)$ and $K_i(t)$ can be random. More specifically, the randomness is given by $$\xi(t)=(\alpha_i(t), K_i(t), d_{ij}(t))$$ which is assumed to form an iid sequence in $(0,\infty)^n \times (0,\infty)^n \times [0,1]^{n^2}\subset\R^\ell$ where $\ell=2n+n^2$.

The following assumptions can be explained very similarly to the ones we made in the Beverton-Holt model.
\begin{asp}\label{a:hassell2}
The following hold:
\begin{enumerate}
    \item $c\in (0,1)$.
    \item $\E\left(\ln^+\left(\max_j \frac{\alpha_j}{K_j^c}+ \max_j \alpha_j +  \max_j \alpha_jK_j^2\right)\right)<\infty$.
    \item With probability one: $\alpha_i(1)>0, K_i(1)> 0, 0<d_{ij}(1)<1$ and $\sum_{j=1}^nd_{ij}(1)=1$ for $i,j=1,\dots, n$. 
\end{enumerate}
\end{asp}
\begin{rmk}
The assumptions are similar to the ones from Assumption \ref{a:BH_dd}. Assumption (1) is needed in order to ber able to use Theorem \ref{t:dispersal_det} - see Remark \ref{r:func} and also .     
\end{rmk}

\begin{thm}\label{t:stoc_m_h}
Suppose the dynamics is given by \eqref{e:disper_hass} and Assumption \ref{a:hassell2} holds. Then the Lyapunov exponent $M$ exists. If $M<0$ then with probability one $\NN(t)\to 0$ as $t\to \infty$ for any $\NN(0)\in\R_+^n$.  If $M>0$, then there exists a positive random vector $\NN(\infty)$ and $\NN(t)\to \NN(\infty)$ in distribution as $t\to\infty$. The total expected population size at stationarity will then be given by
 \[
\sum_{i=1}^n\E N_i(\infty).
 \]
 Moreover, if we look at the deterministic system
 \begin{equation*}
\bar N_i(t+1) =  \sum_{j} \bar d_{ij} \bar f_j(\bar N_j(t))\bar N_j(t), ~i=1,\dots,n
\end{equation*}
where $\bar f_i(x) = \frac{\bar \alpha_i}{(1+\bar K_i x)^c}$ and $\bar \alpha_i = \E \alpha_i(1), \bar K_i = \E K_i(1), \bar d_{ij} = \E d_{ij}(1)$, then this system has a unique globally attracting fixed point $(\bar N_1,\dots, \bar N_n)$ and the total change in expected population size due to environmental fluctuations is well-defined and is given by
\[
H= \sum_{i=1}^n\E N_i(\infty) -  \sum_{i=1}^n \bar N_i.
\]
 
\end{thm}
The proof of this result appears in Appendix \ref{a:hassell}.

\subsection{Slow dispersal in the Hassell $n$ patch model}\label{s:SF_H}
We do similar things to those from Section \ref{s:SF}. Suppose we simplify things by assuming that at each time step a fraction $\eta\in[0,1]$ of the population disperses and dispersed individuals are equally likely to pick any of the $n$ patches. We then get

\begin{equation} \label{e:BV_disper_stoc_unif_H}
\NN(t+1) =\left([(1-\eta)I_n + \eta U_n]\diag\left(  \frac{\alpha_1}{(1+K_1(t)N_1(t))^c},\dots,\frac{\alpha_n}{(1+K_1(t)N_n(t))^c}\right)\right)\NN(t)
\end{equation}
In in the small dispersal limit $\eta=0$ one can find \citep{BS09} the Lyapunov exponent explicitly, namely, 
\begin{equation*}
    M(0)=\max_i\E \ln \alpha_i(0).
\end{equation*}
For brevity, we do not repeat all the details as they are similar to those from Section \ref{s:SF}.
If $M>0$ total population size at stationarity is given by
\[
\sum_i\E N_i(\infty). 
\]
Assume $K_i(t)$ is random and forms an iid sequence $(K_i(t))_{t\geq 1}$.  Let $\bar K_i = \E K_i(1)$ and assume that $\alpha_i>1$. 
The deterministic dynamics is
\begin{equation}\label{e:has}
    N_i(t)=\frac{\alpha_i N(t)}{(1+\bar K_i N(t))^c}.
\end{equation}
The equation has an asymptotically stable fixed point $\bar N_i=\frac{\alpha_i^{1/c}-1}{\bar K_i}>0$ if and only if $\alpha_i>1$ and $0<c\leq 1$. Under these assumptions, as $t\to \infty$ one has $N_i(t)\to \bar N_i$. 

If $r(\delta_0)=\ln \alpha_i >0$ using the results from \cite{HS25a} one can show that the process $(N_i(t))$ has a unique stationary distribution $\mu_i$ on $\R_{++}$. It can then be shown \citep{HS25a} using Birkhoff's ergodic theorem that 
\[
\E \ln\left[\frac{\alpha }{(1+K_i(1) N_i(\infty))^c} \right] = 0.
\]
Using Jensen's inequality
\[
\ln \left[\alpha^{\frac{1}{c}}\right] = \E \ln (1+ K_i(1) N_i(\infty))< \ln (1 + \E K_1 \E N_i(\infty)) =  \ln (1 + \bar K_i \E N_i(\infty)). 
\]
As a result, contrary to what we got in the Beverton-Holt setting when the carrying capacity was fluctuating, we get
\[
\bar N_i = \frac{(\alpha)^{\frac{1}{c}}-1}{\bar K_i} < \E N_i(\infty).
\]
The noise always \textit{increases} the total population size in patch $i$ so that
\[
\sum_i \bar N_i < \sum_i \E N_i(\infty,\eta)
\]
for all $\eta$ close enough to $0$, where we sum over all the $i$ which are source patches i.e. satisfy $m_i>1$. This shows that in this setting for \textit{small dispersal and any type of noise, the total expected population size at stationarity is strictly larger than the total population size without noise.}

\section{Markovian environmental forcing}
We present here some interesting results in a slightly different context, which fits in the general framework from \cite{BS09, RS14, BS19, HNC21}. Suppose $Y_t$ is an irreducible discrete-time Markov chain on the finite state space $I:=\{1,\dots,n_0\}$ and the transition probabilities are given by $\PP (Y_{t+1}=v~|~Y_t=u)=p_{uv}$. Note that $Y_t$ will have a uniqe stationary distribution $(\nu_1,\dots,\nu_{n_0})$ which can be computed explicitly in function of the transition rates $p_{uv}$.

Instead of \eqref{e:dispersal_stoch} assume that the stochastic dynamics is given by

\begin{equation}\label{e:dispersal_PDMP}
N_i(t+1) =  \sum_{j=1}^n d_{ij}(Y_t) f_j(N_j(t),Y_t)N_j(t), ~i=1,\dots,n
\end{equation}

Assume that in each fixed environment $u$ the deterministic dynamics 
\[
N^u_i(t+1) =  \sum_{j=1}^n d_{ij}(u) f_j(N^u_j(t),u)N^u_j(t), ~i=1,\dots,n
\]
satisfies the assumptions of Theorem \ref{t:dispersal_det} and that $\NN^u(t)\to \bar \NN^u\geq 0$ as $t\to\infty$. 

\textbf{Slow switching.} Assume that $p_{uv}$ is very small for all $u,v$, in the sense that $p_{uv} \approx 0$. Then the process spends a long time in state $u$ before it jumps to another state $v$, where it also spends a long time before it jumps to another state and so on. It therefore makes sense to expect that the process will spend a fraction of time $\nu_u$ near $\bar \NN^u$ for all $u\in I$ so that 
\begin{equation}
    \E N_i(\infty) \approx \sum_{u =1}^{n_0} \nu_u \bar N^u_i
\end{equation}
which implies
\begin{equation}
    \sum_{i=1}^n \E N_i(\infty) \approx \sum_{i=1}^n\sum_{u =1}^{n_0} \nu_u \bar N^u_i = \sum_{u =1}^{n_0} \nu_u \sum_{i=1}^n \bar N^u_i.
\end{equation}
Let $\bar N^u_\sigma= \sum_{i=1}^n \bar N^u_i$ be the total population in environment $u$. Then we get
\[
\sum_{i=1}^n \E N_i(\infty) = \sum_{u =1}^{n_0} \nu_u \bar N^u_\sigma\leq \sum_{u =1}^{n_0} \nu_u \max_u \bar N^u_\sigma = \max_u \bar N^u_\sigma   \sum_{u =1}^{n_0} \nu_u = \max_u \bar N^u_\sigma. 
\]
The expectation of the total population will be bounded above by the largest deterministic total population. Similar computations yield the lower bound 
\[
 \min_u \bar N^u_\sigma\leq \sum_{i=1}^n \E N_i(\infty).
\]

If there is only one patch and the functional response is of Beverton-Holt type one can see that $N^u=\frac{r(u)-1}{\kappa}$ so that
\begin{equation}\label{e:N_sp_slow}
\E N_s(\infty)= \nu_1 \frac{r(1)-1}{\kappa} + \nu_2 \frac{r(2)-1}{\kappa}.
\end{equation}

\textbf{Fast switching.} Suppose we want to study the other extreme, when $p_{ij}\approx 1$ so that at every time step there is switching. Assume for simplicity here that $I=\{1,2\}$ and assume at first there is only one patch. The dynamics can be approximated by a periodic difference equation where at each time step the vector field switches from $F$ to $G$, so for example
\begin{eqnarray*}
N_f(t+1)&=&G(N_f(t))\\
N_f(t+2)&=&F(G(N_f(t))\\
N_f(t+3)&=&G(F(G(N_f(t)))    
\end{eqnarray*}
Set $H_1:= G\circ F$ and $H_2:=F\circ G$. Assume that
\[
Y(t+1) = H_1(Y(t))
\]
has a globally stable equilibrium $Y_\infty$ and
\[
Z(t+1) = H_2(Z(t))
\]
has a globally stable equilibrium $Z_\infty$. As expected, one will have the relationships
\[
Z_\infty = F(Y_\infty), Y_\infty = G(Z_\infty).
\]
Then it is reasonable to approximate the distribution of $X(\infty)$ with the distribution:
\[
\frac{1}{2}\delta_{Y_\infty} + \frac{1}{2}\delta_{Z_\infty}.
\]
As a result
\[
\E N_f(\infty) = \frac{1}{2} Y_\infty + \frac{1}{2}Z_\infty.
\]
For a simple example assume that we switch between two Beverton-Holt dynamics i.e. the stochastic model is
\[
N_f(t+1) = \frac{r(t)N_f(t)}{1+\kappa N_f(t)} N_f(t)
\]
Then the above results will imply that
\begin{equation}\label{e:N_sp_fast}
\E N_f(\infty) = \frac{1}{2} Y_\infty + \frac{1}{2}Z_\infty = \frac{1}{2} \frac{r(1)r(2)-1}{\kappa(1+r(2))} + \frac{1}{2}\frac{r(1)r(2)-1}{\kappa(1+r(1))}. 
\end{equation}

Note that at least in the single patch case we can see that slow switching is always better than fast switching as by \eqref{e:N_sp_fast} for $\nu_1=\nu_2=\frac{1}{2}$ and by \eqref{e:N_sp_slow} we have
\begin{eqnarray*}
\E N_s(\infty) - \E N_f(\infty) &=&  \frac{1}{2}\frac{r(1)-1}{\kappa} + \frac{1}{2} \frac{r(2)-1}{\kappa} - \left(\frac{1}{2} \frac{r(1)r(2)-1}{\kappa(1+r(2))} + \frac{1}{2}\frac{r(1)r(2)-1}{\kappa(1+r(1))}\right)\\
&=& \frac{1}{2}\frac{(r(1)-r(2))^2}{\kappa (1+r(1))(1+r(2))}\\
&>&0.
\end{eqnarray*}

We can lift up this single-patch result to the multipatch setting as follows. In the small dispersal limit $\eta\approx 0$ we get that for slow switching 
\[
\sum_{i=1}^n \E N^s_i(\infty,\eta)\approx  \sum_{i=1}^n \sum_u \nu_u \frac{r_i(u)-1}{\kappa_i}
\]
while in the slow switching case if $I=\{1,2\}$ we get 
\[
\sum_{i=1}^n \E N^f_i(\infty,\eta)\approx\sum_{i=1}^n\frac{1}{2} \frac{r_i(1)r_i(2)-1}{\kappa_i(1+r_i(2))} + \frac{1}{2}\frac{r_i(1)r_i(2)-1}{\kappa_i(1+r_i(1))}.
\]
As a result, for small dispersal and two environmental states we have
\begin{eqnarray*}
\sum_{i=1}^n (\E N_i^s(\infty) - \E N_i^f(\infty)) &=& \sum_{i=1}^n\frac{1}{2}\frac{(r_i(1)-r_i(2))^2}{\kappa_i (1+r_i(1))(1+r_i(2))}\\
&>&0,
\end{eqnarray*}
which shows that slow switching environments are always better than fast switching environments and lead to a higher total population size. 

\textbf{Biological interpretation:} The slow/fast switching of the Markov chain can be related to the temporal correlations between the various environmental states. If the switching is fast $p_{uv}\approx 1$ then the environments are negatively correlated (correlation is $-1$) while if the switching is slow $p_{uv}\approx 0$ the environments are positively correlated (correlation is $+1$). Positive auto correlations give the population a better chance to grow as it has time to adapt to the environment and fully use its resources. Negative autocorrelations make things unpredictable and the population size decreases.

\section{Small noise approximations}\label{s:small_noise}

It is usually hard or impossible to get information about the invariant measure $\mu$ which describes the persistence of the multi-patch system. The general theory can tell us that $\mu$ exists and is unique but we don't have any formulas describing this measure, and even computing the population sizes at stationarity is usually not possible. In order to circumvent this problem, one can use small perturbations of a deterministic system which has `nice' fixed points. 
Assume we have the deterministic dynamics given by

\begin{equation}\label{e:SDE_det}
\bar \NN(t+1) = F(\bar \NN(t), \bar e).
\end{equation}

We perturb the dynamics by letting the $(e(t))_{t\geq 0}$ be iid random variables taking values in a convex compact subset of $\R^\ell$. The stochastic dynamics is
\begin{equation}\label{e:gen2}
N_i(t+1) = F_i(\NN(t), e(t)).
\end{equation}
We note that the systems \eqref{e:dispersal_det2} and \eqref{e:dispersal_stoch2} can be put in the forms \eqref{e:SDE_det} and \eqref{e:gen2}

The next remark provides some intuition and heuristics. 
\begin{rmk} Using  the theory of stochastic persistence one can show, under certain assumptions, that the system population given by \eqref{e:SDE_det} persists and converges to a unique invariant probability measure (stationary distribution) $\mu$. There are few tools that can give us the analytical properties of the stationary distribution. This makes it hard to quantify how noise changes the population size at stationarity, since this is a functional of the stationary distribution, i.e. the $i$th patch has the expected population size $\E N_i(\infty)=\int x_i \mu(dx)$. One can sidestep this issue by modifying the approximation methods developed by \cite{stenflo1998ergodic, cuello2019persistence}. Cuello's results \citep{cuello2019persistence} lead to a Taylor-type approximation which says loosely that if the noise is small, that is $e=\bar e + \rho \xi$ for $0<\rho \ll 1$, and the deterministic dynamics $\bar\NN(t+1) = F(\bar\NN(t),\bar e)$ has a stable equilibrium $\bar \NN_0$ then for initial conditions $\NN(0)$ close to that equilibrium the dynamics converges as $t\to\infty$ to a stationary distribution $\mu$ that lives `around' $\bar \NN_0$. If $\bar \NN_0$ is globally stable then one can show, under natural assumptions, that the convergence holds for all nonzero initial conditions. \cite{cuello2019persistence} then uses these results to prove that $\E_\mu N_i$ has a Taylor-type expansion that allows one to compute all the terms explicitly as a function of $F$ and its various derivatives, including derivatives with respect to noise terms, evaluated at $(\bar \NN_0, \bar e).$  
\end{rmk}

Using the results from \cite{stenflo1998ergodic} and especially from \cite{cuello2019persistence} (see Appendix \ref{a:small_noise} for details and assumptions) in our setting we get the following result.

Define $vec(\cdot)$ to be the operation that concatenates the columns of a matrix into a vector and let $\otimes$ denote the Kronecker product of matrices.

\begin{thm}\label{t:small_noise}
Suppose $e=\bar e + \rho \xi(t)$ lives in a compact convex subset of $\R^\ell$ and the assumptions of Theorem \ref{t:dispersal_det} hold for $f_i(\cdot,\bar e), \bar D:=D(\bar e)$. Furthermore, suppose that $\lambda(D(\bar e)\Lambda(0,\bar e))>0$.  Then the deterministic system \eqref{e:dispersal_det} persists and has a unique globally attracting equilibrium $\bar N_0$. 

Under natural assumptions on the noise  (see Appendix \ref{a:small_noise}) the stochastic system \eqref{e:dispersal_stoch} has a unique invariant measure $\mu$ such that $\bar N_0\in \supp(\mu)$ and for any $\NN(0)$ one has $\NN(t)\to \NN$ in distribution where $N$ is distributed according to $\mu$. 

Let $D_N F$ be the Jacobian at $\bar N_0$ with respect to the population i.e. the $i-j$th entry of $D_N F$ is $\frac{\partial F_i}{\partial N_j}(\bar N_0)$. Similarly, $D_{\xi} F$ will be the matrix with entries $\diag(\frac{\partial F_i}{\partial \xi_i}(\bar N_0))$.

Assume for simplicity that in the deterministic setting we set the dispersal to zero  so that $\bar N_0$ is the dispersal-free equilibrium which is assumed to be globally asymptotically stable.

We have 
\begin{equation}\label{e:approx}
\begin{split}
\E(N)=&\bar N_0 + \rho BD_{\xi}F(\bar N_0,0)(\E(\xi))+ \frac{1}{2}B\left(\sum_{i,j=1}^n \frac{\partial^2 F}{\partial N_i\partial N_j}(\bar N_0,0)[\text{Cov}(\NN)]_{ij}\right)\\
+&\frac{\rho^2}{2}B\left(\frac{\partial^2 F}{\partial \xi_i\partial \xi_j}[\text{Cov}(\xi)]_{ij}\right)+ \frac{\rho^2}{2}B\left(\sum_{i,j=1}^n\frac{\partial^2 F}{\partial \xi_i\partial \xi_j}(\bar N_0,0)\E(\xi)_i\E(\xi)_j\right)\\
+&\frac{\rho^2}{2}B\left(\sum_{i,j=1}^n \frac{\partial^2 F}{\partial N_i\partial N_j}(\bar N_0,0)(BD_{\xi}F(\E\xi))_i(B D_{\xi}F(\E\xi))_j \right)\\
+&\rho^2B\sum_{i=1}^n\sum_{j=1}^n\frac{\partial^2 F}{\partial N_i\partial \xi_j}(\bar N_0,0)(\E(\xi))_j(BD_{\xi}F(\E\xi))_i+ O(\rho^3)    
\end{split}
\end{equation}
where $B$ denotes the matrix $B:=\left(I-D_NF(\bar N_0)\right)^{-1}$ and $[\text{Cov}(\NN)]_{ij}$ the entries of the matrix defined by
\[
vec(\text{Cov}(\NN)) = (I-D_NF \otimes D_NF)^{-1} (D_\xi F \otimes D_\xi F) vec(Cov(\xi))
\]
The change in total expected population size because of noise is therefore given by
\begin{equation}\label{2ordApp}
\begin{split}
\E \NN - \bar N_0 =&\rho BD_{\xi}F(\bar N_0,0)(\E(\xi))+ \frac{1}{2}B\left(\sum_{i,j=1}^n \frac{\partial^2 F}{\partial N_i\partial N_j}(\bar N_0,0)[\text{Cov}(\NN)]_{ij}\right)\\
+&\frac{\rho^2}{2}B\left(\frac{\partial^2 F}{\partial \xi_i\partial \xi_j}[\text{Cov}(\xi)]_{ij}\right)+ \frac{\rho^2}{2}B\left(\sum_{i,j=1}^n\frac{\partial^2 F}{\partial \xi_i\partial \xi_j}(\bar N_0,0)\E(\xi)_i\E(\xi)_j\right)\\
+&\frac{\rho^2}{2}B\left(\sum_{i,j=1}^n \frac{\partial^2 F}{\partial N_i\partial N_j}(\bar N_0,0)(BD_{\xi}F(\E\xi))_i(B D_{\xi}F(\E\xi))_j \right)\\
+&\rho^2B\sum_{i=1}^n\sum_{j=1}^n\frac{\partial^2 F}{\partial N_i\partial \xi_j}(\bar N_0,0)(\E(\xi))_j(BD_{\xi}F(\E\xi))_i+ O(\rho^3)    
\end{split} 
\end{equation}
\end{thm}
\begin{proof}
The result is proved in the Appendix -- see Theorem \ref{t:small_noise_A}. 
\end{proof}
\begin{rmk}
Note that we made the assumption for the approximation formula \eqref{e:approx} that $\bar N_0$ comes from the system without any dispersal. This is for two reasons. First, in most applications $\bar N_0$ can only be computed explicitly when there is no dispersal. Second, the equation \eqref{e:approx} is quite involved already so we wanted to keep things as simple as possible. Formulas can be found of course even in the case where the equilibrium $\bar N_0$ comes from a system where there is non-zero dispersal.
\end{rmk}

\textbf{Biological interpretation:} If one perturbs a deterministic ecological system that has a globally attracting fixed point $\hat \NN^{\bar D}$ by small noise which fluctuates like $\rho \xi(t)$ around the mean $\bar e$ for a small $\rho>0$, then the stochastic system converges to an invariant probability measure supported around $\hat \NN^{\bar D}$. Furthermore, the expectation of the $i$th population size at stationarity looks like $\hat N_i$ plus a term of order $\rho^2$. This $\rho^2$ term  depends on the first and second derivatives with respect to the population sizes and random parameters of the vector field evaluated at the deterministic fixed point $(\hat \NN^{\bar D}, \bar e)$ as well as on the covariance of the noise terms. This shows that one can do a Taylor-type expansion around the deterministic fixed point and that, as expected, small environmental fluctuations change the dynamics in a `smooth' fashion.

\begin{rmk}
Small noise approximations have been used in related settings before. Some of these approximations are similar to the $\delta$ method from statistics \citep{oehlert1992note}. Small noise approximations were employed for various epidemic and population dynamics models \citep{bartlett1956deterministic, bartlett1957theoretical} as well as in fisheries models \citep{horwood1983general, getz1984production}.
\end{rmk}

\subsection{Small noise approximations for a Beverton-Holt $2$-patch model}
We will focus on the $n=2$ patch case and do the approximations explicitly. 
Let superscripts, like $N^{\delta, \rho}$, denote the dependence on the dispersal matrix and the noise intensity $\rho$. According to Proposition \ref{p:grumbach} we look at the case when $r_1,r_2>1$, and $r_1\neq r_2$. Then, in the deterministic setting from \eqref{e:det_disc} we get by \cite{grumbach2023effect} that
$$H(\delta):=(N_1^{\delta, 0}(\infty) +N_2^{\delta,0}(\infty))-(N_1^{0,0}(\infty)+N_2^{0,0}(\infty))=0$$
has two possible solutions. One solution is for $\delta=0$ and one gets
\[
(N_1^{\delta, 0}(\infty),N_2^{\delta, 0}(\infty)) = (N_1^{0, 0}(\infty),N_2^{0, 0}(\infty)) = (K_1, K_2).
\]
The other solution is for
$$\delta=\Tilde{\delta}= \frac{K_1K_2(r_1-1)(r_2-1)(r_1-r_2)}{(K_1(r_2-1)+K_2(r_1-1))(K_1r_1(r_2-1)-K_2r_2(r_1-1)}$$
and this yields
$$(N_1^{\delta, 0}(\infty),N_2^{\delta, 0}(\infty)) = (N_1^{0, 0}(\infty),N_2^{0, 0}(\infty))=\left(\frac{K_1(K_1+K_2)(r_2-1)}{K_1(r_2-1)+K_2(r_1-1)}, \frac{K_2(K_1+K_2)(r_1-1)}{K_1(r_2-1)+K_2(r_1-1)}\right).$$

We will next use Theorem \ref{t:small_noise} to compute the expected total population when there are small environmental fluctuations.

\textbf{Case 1 -- intrinsic growth rates are random}:
Suppose $\delta=0$, and $r_i(t)=\bar r_i+\rho\xi_i(t)$ where  $\xi(t)=(\xi_1(t),\xi_2(t))$ are i.i.d. random vectors, with $\E(\xi_i(t))=0$ and $\E(\xi^2_i(t))=1$. One can see that
  
$$A=\begin{bmatrix}
  \frac{K_1^2\rho^2}{\bar r_1^2 - 1} & (\frac{K_1K_2}{\bar r_1\bar r_2-1})\E(\xi_1\xi_2)
 \\
 (\frac{K_1K_2}{\bar r_1\bar r_2-1})\E(\xi_1\xi_2)    &  \frac{K_2^2\rho^2}{\bar r_2^2 - 1}
 \end{bmatrix}
$$

and
$$B= \begin{bmatrix}
  \frac{\bar r_1}{\bar r_1 - 1} & 0
 \\
  0 &  \frac{\bar r_2}{\bar r_2 - 1}
 \end{bmatrix}. 
$$
We have 
$$\E(N_i^{0,\rho})= K_i - \frac{2K_i\rho^2}{\bar r_i(\bar r_i^2-1)}+ O(\rho^3)$$
and hence the expectation of the total population size is
\[
\E(N_1^{0,\rho})+\E(N_2^{0,\rho})=K_1+K_2-\frac{2K_1\rho^2}{\bar r_1(\bar r_1^2-1)}-\frac{2K_2\rho^2}{\bar r_2(\bar r_2^2-1)}+O(\rho^3).
\]

\textbf{Biological Interpretation}: \textit{Since the patches are isolated, $\delta=0$, we can see that, we are simply looking at the 1D Beverton-Holt model, with linear noise in $r$. The expected population is less than the deterministic total population size at equilibrium. }

\textbf{Case 2 -- dispersal rate and carrying capacities are random }:

Let $\delta(\xi)=\rho\xi_1$ and $K_1=\overline{K}_1+\rho\xi_2$, and $K_2=\overline{K}_2+\rho\xi_2$ where $(\xi(t))_{t=1}^{\infty}$ are iid random vectors, and there exists $C>0$ such that $\|\xi(t)\|\leq C$ . It is natural to assume $\E \xi_2=0$ and in order to have some dispersal we need $\E \xi_1>0$. If we set $\alpha=\frac{\bar r_2(\bar r_1-1)}{\bar r_1(\bar r_2-1)}$ we get the following expression by Theorem \ref{t:small_noise} 
\begin{equation*}\label{pop_exp_2}
    \begin{split}
\E(N_1(\infty))+\E(N_2(\infty))&=\overline{K}_1+\overline{K}_2+\rho(\overline{K}_2-\overline{K}_1)\E(\xi_1)\left(\frac{\bar r_1}{\bar r_1-1}-\frac{\bar r_2}{\bar r_2-1}\right) \\
&-\frac{2}{\bar r_1\overline{K}_1}\rho^2\left( \frac{\bar r_1-1}{\bar r_1+1}+\frac{\bar r_1^2(\bar K_1-\bar K_2)^2}{(\bar r_1^2-1)}+1\right)\\
&-\frac{2}{\bar r_2\overline{K}_2}\rho^2\left(\frac{\bar r_2-1}{\bar r_2+1}+\frac{\bar r_2^2(\bar K_1-\bar K_2)^2}{(\bar r_2^2-1)}+1\right)\\
&+2\rho^2\left(\alpha+\frac{1}{\alpha}-2\right)\left(\E(\xi_1\xi_2)\right)\\
&+4(\bar K_1-\bar K_2)\rho^2\left(\frac{1}{\bar K_1(\bar r_1+1)}-\frac{1}{\bar K_2(\bar r_2+1)}\right)\E(\xi_1\xi_2)\\
&+\rho^2\left(\frac{1}{r_1-1}-\frac{1}{r_2-1}\right)\left(\frac{r_1}{r_1-1}-\frac{r_2}{r_2-1}\right)(K_1-K_2)(\E(\xi_1))^2\\
&-\rho^2\left(\frac{r_1-1}{K_1}\right)\left((K_2-K_1)\left(\frac{r_1}{r_1-1}\right)\E(\xi_1)\right)^2\\
&-\rho^2\left(\frac{r_2-1}{K_2}\right)\left((K_1-K_2)\left(\frac{r_2}{r_2-1}\right)\E(\xi_1)\right)^2\\
&+O(\rho^3).
\end{split}
\end{equation*}

\textbf{Biological Interpretation}: \textit{From the above expression we can observe, since $x\mapsto \frac{x}{x-1}$ is decreasing, that if we have $K_1>K_2$ and $r_2>r_1$ or $K_2>K_1$ and $r_1>r_2$ then the expected total population size is greater than the deterministic total population size at equilibrium for small enough $\rho$. This shows that adding some `random' dispersal can increase the total population size. The second order terms are significantly more complicated and depend not only on the coefficients of the model but also on the covariance $\E \xi_1\xi_2$.}

\textbf{Case 3 -- dispersal rate and carrying capacities are random }: Let
$\delta(\xi)=\rho\xi_1$,  $K_1=\overline{K}_1+\rho\xi_2$, and $K_2=\overline{K}_2+\rho\xi_3$ where $(\xi(t))_{t=1}^{\infty}$ are iid random vectors with $\E(\xi_1)>0$, $\E(\xi_2)=\E(\xi_3)=0$ and $\E(\xi_i^2)=1$. We set once again $\alpha=\frac{r_2(r_1-1)}{r_1(r_2-1)}$. Using Theorem \ref{t:small_noise} we get the small noise approximation
\begin{equation*}\label{pop_exp_3}
\begin{split}
\E(N_1(\infty))+\E(N_2(\infty))&=\overline{K}_1+\overline{K}_2+\rho(K_2-K_1)\left(\frac{r_1}{r_1-1}-\frac{r_2}{r_2-1}\right)\E(\xi_1)\\
&-\frac{2}{K_1r_1}\rho^2\left(\frac{r_1-1}{r_1+1}+\frac{r_1^2(\overline{K_1}-\overline{K_2})^2}{(r_1^2-1)}+ 1\right)\\
&-\frac{2}{r_2K_2}\rho^2\left(\frac{r_2-1}{r_2+1}+\frac{r_2^2(\overline{K}_1-\overline{K}_2)^2}{(r_2^2-1)}+ 1 \right)+2\rho^2\left(\frac{2(\overline{K_1}-\overline{K_2})}{(r_1+1)\overline{K_1}}+\alpha-1\right)\E(\xi_1\xi_2)\\
&+2\rho^2\left(\frac{1}{\alpha}-\frac{2(\overline{K}_1-\overline{K}_2)}{(r_2+1)K_2}-1\right)\E(\xi_1\xi_3)\\
&-\left(\frac{r_1^2}{K_1(r_1-1)^2}+\frac{r_2^2}{K_2(r_2-1)^2}\right)(K_1-K_2)^2(\E(\xi_1)^2\\
&+\rho^2(K_1-K_2)\left( \frac{r_1}{(r_1-1)^2}+\frac{r_2}{(r_2-1)(r_1-1)}\right)(\E(\xi_1))^2\\
&+\rho^2(K_1-K_2)\left( \frac{r_2}{(r_2-1)^2}+\frac{r_1}{(r_2-1)(r_1-1)}\right)(\E(\xi_1))^2\\
&+O(\rho^3).
\end{split}
\end{equation*}

\textbf{Biological Interpretation:} \textit{Since the first order term in \eqref{pop_exp_3} is the same as the one in \eqref{pop_exp_2}, hence if noise is small the fluctuations in the carrying capacities do not affect the expected population size.}

\subsubsection{Heuristic proof for small noise and small dispersal}
We give an intuitive explanation of the proof. We can write the model from \cite{grumbach2023effect} as

\begin{align}
    N(t+1)=\begin{bmatrix}
     1-\delta  &  \delta     \\  
     \delta       &       1-\delta             
    \end{bmatrix} \begin{bmatrix}
    \frac{r_1(\xi_1)}{1+\kappa_1 N_1(t)}        &  0    \\  
     0                &            \frac{r_2(\xi_2)}{1+\kappa_2N_2(t)}       
    \end{bmatrix}\begin{bmatrix}
    N_1(t)           \\  
    N_2(t)         \\                      
    \end{bmatrix}
\end{align}

\begin{align}
    N(t+1)=\begin{bmatrix}
     \frac{(1-\delta)r_1(\xi_1)N_1(t)}{1+\kappa_1N_1(t)} + \frac{\delta r_2(\xi_2)N_2(t)}{1+\kappa_2 N_2(t)}                       \\  
      \frac{\delta r_1(\xi_1)N_1(t)}{1+\kappa_1 N_1(t)}  +  \frac{(1-\delta)r_2(\xi_2)N_2(t)}{1+\kappa_2 N_2(t)}                                         \\                    
    \end{bmatrix}
\end{align}
where $\kappa_i=\frac{\overline{r}_i-1}{K_i}$. Since we know that the asymptotic population size with $\delta=0$ and $\rho=0$ is $N^*=(K_1,K_2)$, we rewrite the above equation in the form 

\begin{equation}\label{dpapprox}
    N(t+1)=F(N(t),\delta,\xi(t)).
\end{equation}
In the same spirit as \cite{cuello2019persistence}, consider the first order approximation around $x^*$

\begin{equation}
    Z(t+1)=x^*+ D_xF(Z(t)-x^*)+ \frac{\partial F}{\partial \delta}\delta + D_eF(\xi(t)-\E(\xi)). 
\end{equation}
If we take $\epsilon>0$ sufficiently small and $\rho$ is small enough, if $\|Z(0)-x^*\|<\epsilon$, then we know by Theorem \ref{t:small_noise_A} that $Z(t)\to \overline{Z}$ in distribution when $t\to\infty$. As a result we get by taking expectations and limits in the above equation that

\begin{equation}
    \E(\overline{Z})=x^*+D_xF(\E(\overline{Z})-x^*)+\frac{\partial F}{\partial \delta}\delta.
\end{equation}
If $(I-D_xF)$ is invertible one has

$$\E(\overline{Z})=x^*+(I-D_xF)^{-1}\left(\frac{\partial F}{\partial \delta}\right)\delta.$$
Computing this explicitly gives us
\begin{align}
    \E(\overline{Z})=\begin{bmatrix}
        K_1   \\
        K_2
    \end{bmatrix}
    +
    \begin{bmatrix}
           \frac{r_1}{r_1-1}      & 0                      \\  
                         0        &   \frac{r_2}{r_2-1}                                                
    \end{bmatrix}
    \begin{bmatrix}
             -\delta(K_1-K_2)                             \\  
              \delta(K_1-K_2)                                 
    \end{bmatrix}
=\begin{bmatrix}
    K_1-\frac{r_1\delta(K_1-K_2)}{r_1-1}\\
    K_2+\frac{r_2\delta(K_1-K_2)}{r_2-1} \\
\end{bmatrix}.   
\end{align}
As a result we get an expression of the form
$$E(N_1(\infty)+(N_2(\infty))=x^* -\delta(K_1-K_2)\left(\frac{r_1}{r_1-1}-\frac{r_2}{r_2-1}\right)+O(\rho^2).$$

\textbf{Biological interpretation:} \textit{The environmental fluctuations have no first order contributions to the total population size at stationarity. However, adding dispersal does lead to first order contributions.}

The same type of proof works for the second order approximation 
\begin{equation}
\begin{split}
Z_2(t+1)&=x^*+ D_xF(Z_2(t)-x^*)+ \frac{\partial G}{\partial \delta}\delta + D_eF(\xi(t)-\E(\xi))\\
&+ D_{xx}F(Z_2(t)-x^*,Z_2(t)-x^*)+D_{xe}G(Z_2(t)-x^*,\xi(t)-\E(\xi))\\
&+D_{ee}G(\xi(t)-\E(\xi),\xi(t)-\E(\xi))
\end{split}
\end{equation}
One gets

\begin{equation}
\E N_i=K_i+\frac{r_i\delta(K_1-K_2)(-1)^{i}}{r_i-1}-\frac{2K_i\rho^2}{r_i(r_i+1)(r_i-1)}+O(\rho^3)    
\end{equation}
which then implies
\[
\E N_1 + \E N_2 =K_1+K_2-\delta(K_1-K_2)\left(\frac{r_1}{r_1-1}-\frac{r_2}{r_2-1}\right)-\frac{2K_1\rho^2}{r_1(r_1^2-1)}-\frac{2K_2\rho^2}{r_2(r_2^2-1)}+O(\rho^3).
\]

\textbf{Biological Interpretation}: The above result shows that allowing a small amount of dispersal, like joining two isolated patches by a bottleneck can potentially increase the population size. This has been shown by \cite{grumbach2023effect} analytically in the deterministic case. From the above calculation we can see that if $K_2>K_1$ and $r_1>r_2$, then for $\delta=\rho$ and for $\rho$ sufficiently small, the expected total population size exceeds the deterministic total population at equilibrium for patches. In this setting we also get that environmental fluctuations are always detrimental and decrease the total population size.

\section{Simulations}\label{s:sim}

\subsection{Simulations}
In this section we present simulation results for some of the models we analyzed in the paper. These numerical explorations can be done for both small and large environmental fluctuations and therefore present some interesting phenomena.
For all of the simulations we have taken the random variables to be lognormally distributed. In the cases where a pair of lognormal random variables is considered, we define the correlation matrix $\Sigma$, for example if we have $Y=(Y_1,Y_2)$, then the following relations hold 
\[
\E(Y)_i=e^{\mu_i+\frac{1}{2}\Sigma_{ii}}
\]
\[
Var(Y)_{ij}=e^{\mu_i+\mu_j+\frac{1}{2}(\Sigma_{ii}+\Sigma_{jj})}(e^{\Sigma_{ij}}-1)
\]

We approximate the population size $N(\infty)$ at stationarity by the sample means
\[
\hat N(T) = \frac{1}{T}\sum_{t=0}^T N(t)
\]

All the simulations  have 200 data points (sample means) and each data point comes from $T=500$ iterations of the stochastic difference equation.

\subsection{Two-patch Beverton-Holt Model with dispersal} 

It has been shown by \cite{grumbach2023effect} in the deterministic setting, that one can get four different possible behaviors between the total population size and the dispersal rate $\delta$. The first type of behavior is the \textbf{monotonically beneficial} one where dispersal increases the total population size monotonically. In the \textbf{unimodally beneficial} case dispersal always increases the total population but does this in a unimodal fashion. In the \textbf{beneficial turning detrimental case}, at first dispersal is beneficial to the total population size but after a certain dispersal rate it becomes detrimental and decreases the total population size. Finally, in the \textbf{monotonically detrimental} setting, dispersal is always detrimental and decreases the total population size monotonically.

\begin{rmk}
One of the referees told us that these four responses of the model are not general, and that many other types of responses are possible \citep{franco2024new}.
\end{rmk}

Here, we explore how noise affects each of those behaviors.

\subsubsection{Random Intrinsic Growth Rates}
Here we assume that the growth rates are random and the carrying capacities are fixed. 
\begin{equation}
\begin{aligned}
N_1(t+1)&=\frac{K_1(1-\delta)(s_1(t)+1)(t)N_1(t)}{K_1+s_1(t)N_1(t)} +  \frac{K_2\delta(s_2(t)+1)N_2(t)}{K_2+s_2(t)N_2(t)}\\
N_2(t+1)&=\frac{K_2(1-\delta)(s_2(t)+1)N_2(t)}{K_2+s_2(t)N_2(t)} +  \frac{K_1\delta(s_1(t)+1)N_1(t)}{K_1+s_1(t)N_1(t)}.
\end{aligned}
\end{equation}

The numerical experiments can be seen in Figure \ref{f:BH6}.

\begin{figure}

    \begin{subfigure}{0.45\textwidth}            
            \includegraphics[width=\textwidth]{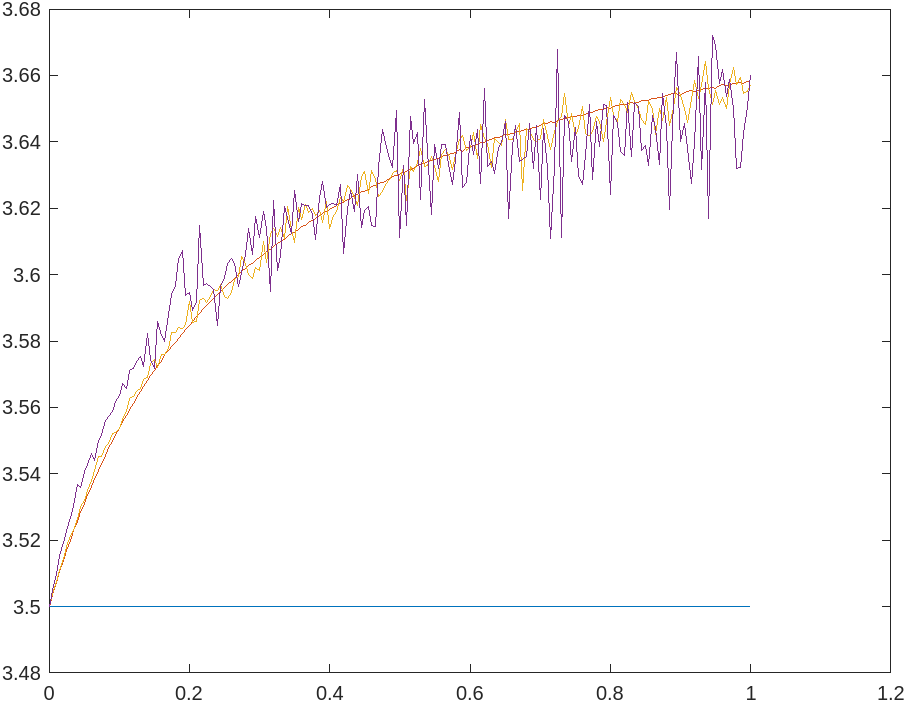}
            \caption{\textbf{Monotonically beneficial}}
           
    \end{subfigure}
    \begin{subfigure}{0.45\textwidth}
            \centering
            \includegraphics[width=\textwidth]{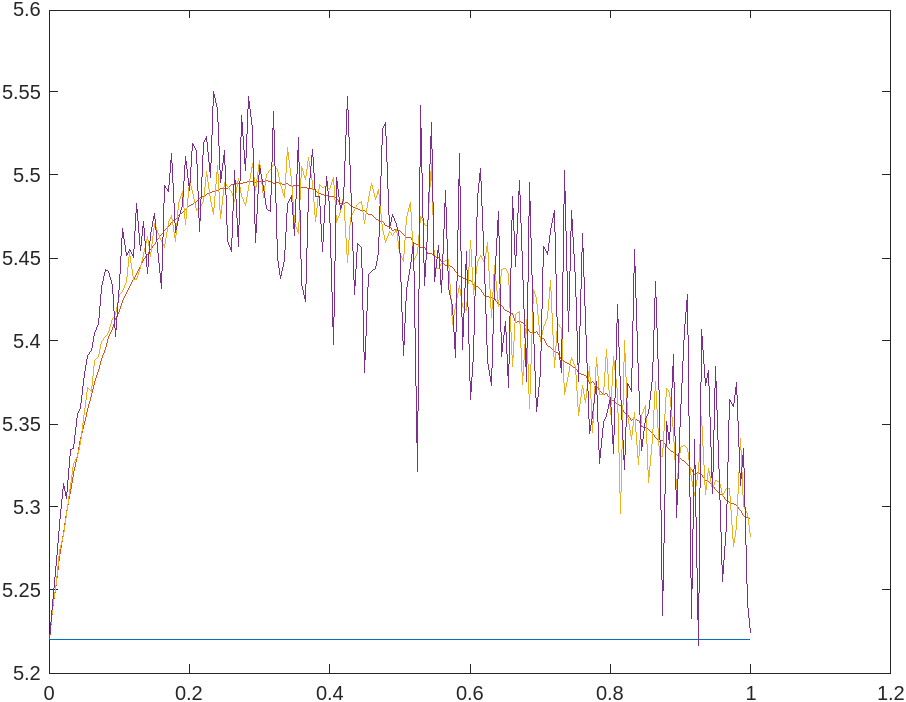}
            \caption{\textbf{Unimodally benefical}}       
    \end{subfigure}
      \begin{subfigure}{0.45\textwidth}
            \centering
            \includegraphics[width=\textwidth]{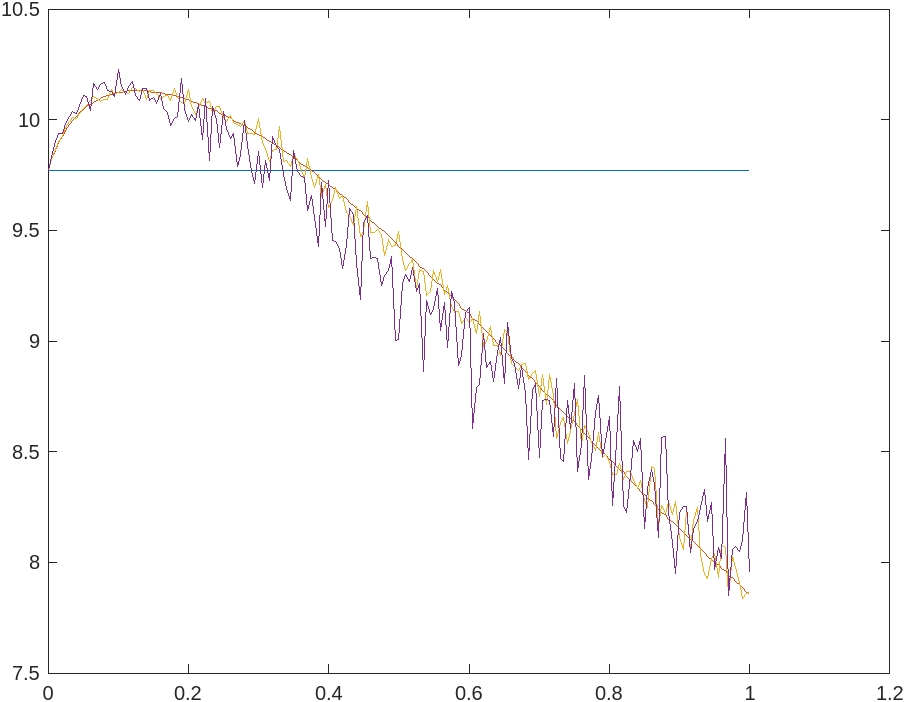}
            \caption{\textbf{Beneficial turning detrimental}}       
    \end{subfigure}
     \begin{subfigure}{0.45\textwidth}
            \centering
            \includegraphics[width=\textwidth]{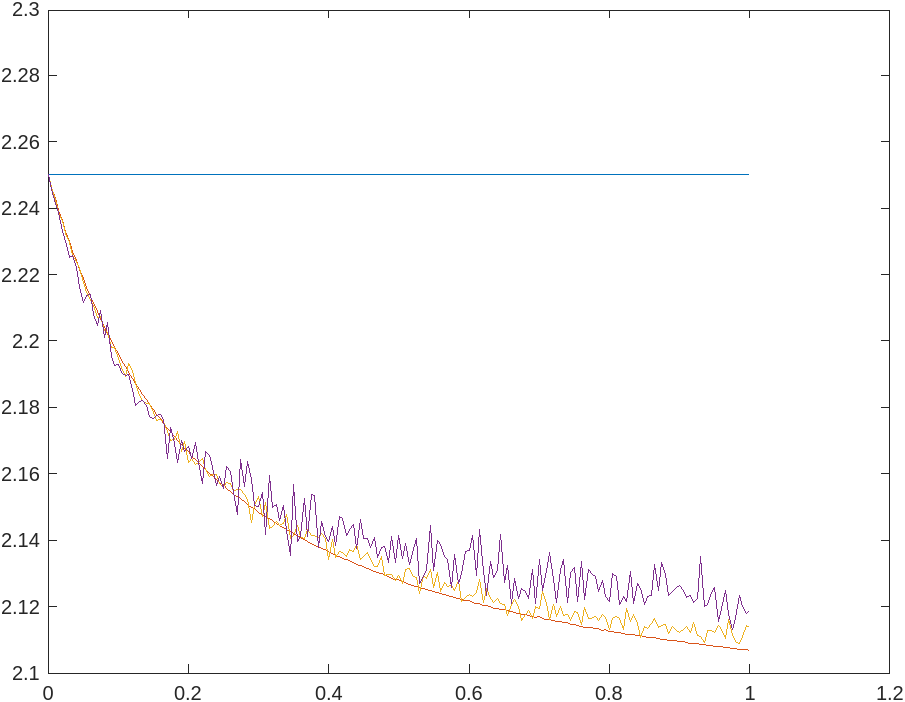}
            \caption{\textbf{Monotonically detrimental}}       
    \end{subfigure}    
    \caption{Total population size versus dispersal rate ($\delta$) in the two-patch Beverton-Holt model, with $c=0.75$. Consider random variables $r_1(t),r_2(t)$ to be uncorrelated: $\Sigma=sI$. \textbf{A}: $r_1=3$, $K_1=2$ and $r_2=1.5$, $K_2=1.5$. \textbf{B}: $r_1=3.2$, $K_1=3.85$ and $r_2=1.5$, $K_2=1.37$. \textbf{C}: $r_1=3.4$, $K_1=8.4$ and $r_2=1.5$, $K_2=1.37$. \textbf{D}: $r_1=2$, $K_1=1$ and $r_2=1.25$, $K_2=1.25$. The red curve corresponds to $s=5\times 10^{-4}$, yellow to $s=0.25$, and the purple one to $s=1.5$. In all four of these plots we observe that adding noise does not seem to change the qualitative relationship between total population and dispersal rate. In the monotonically detrimental case (D), higher noise seems to temper the effect of dispersal, even though it remains monotonically detrimental.}
\label{f:BH6}
\end{figure}

\subsubsection{Random Carrying Capacities}
For these numerical experiments we assume the carrying capacities are random and the growth rates are fixed.
\begin{equation}
\begin{aligned}
N_1(t+1)&=\frac{K_1(t)(1-\delta)r_1N_1(t)}{K_1(t)+(r_1-1)N_1(t)} +  \frac{K_2(t)\delta r_2N_2(t)}{K_2(t)+(r_2-1)N_2(t)}\\
N_2(t+1)&=\frac{K_2(t)(1-\delta)r_2N_2(t)}{K_2(t)+(r_2-1)N_2(t)} +  \frac{K_1(t)\delta r_1N_1(t)}{K_1(t)+(r_1-1)N_1(t)}.
\end{aligned}
\end{equation}

The numerical experiments can be seen in Figure \ref{f:BH7}.

\begin{figure}

    \begin{subfigure}{0.45\textwidth}            
            \includegraphics[width=\textwidth]{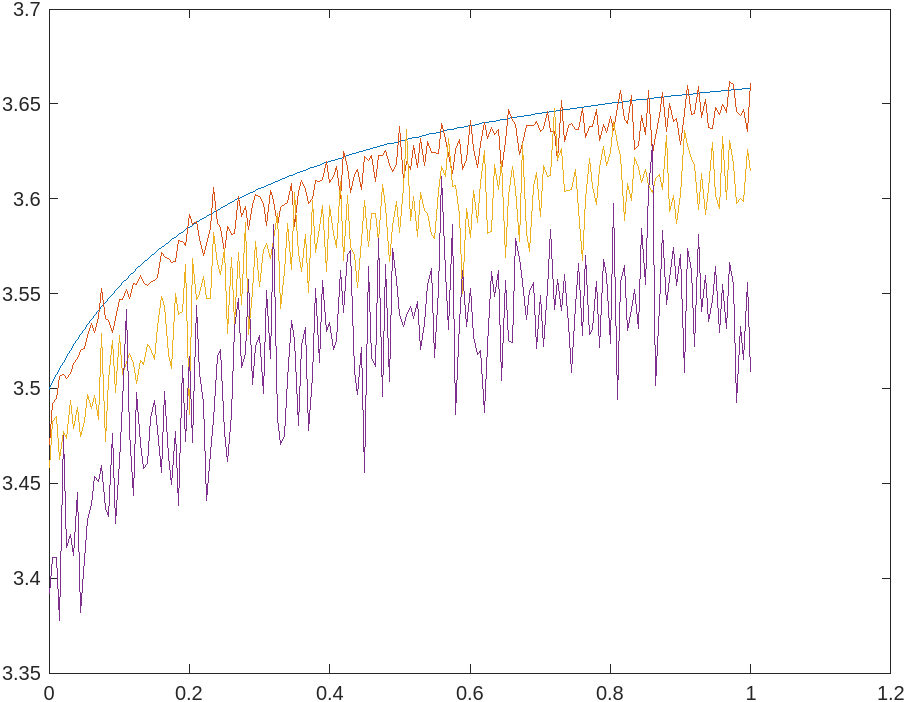}
            \caption{\textbf{Monotonically beneficial}}
           
    \end{subfigure}
    \begin{subfigure}{0.45\textwidth}
            \centering
            \includegraphics[width=\textwidth]{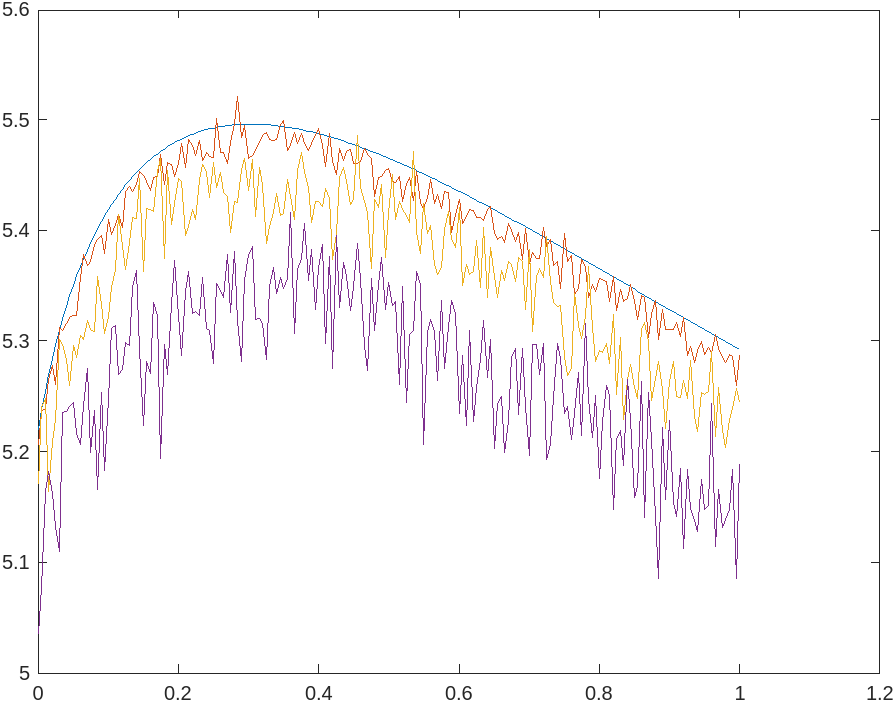}
            \caption{\textbf{Unimodally benefical}}       
    \end{subfigure}
      \begin{subfigure}{0.45\textwidth}
            \centering
            \includegraphics[width=\textwidth]{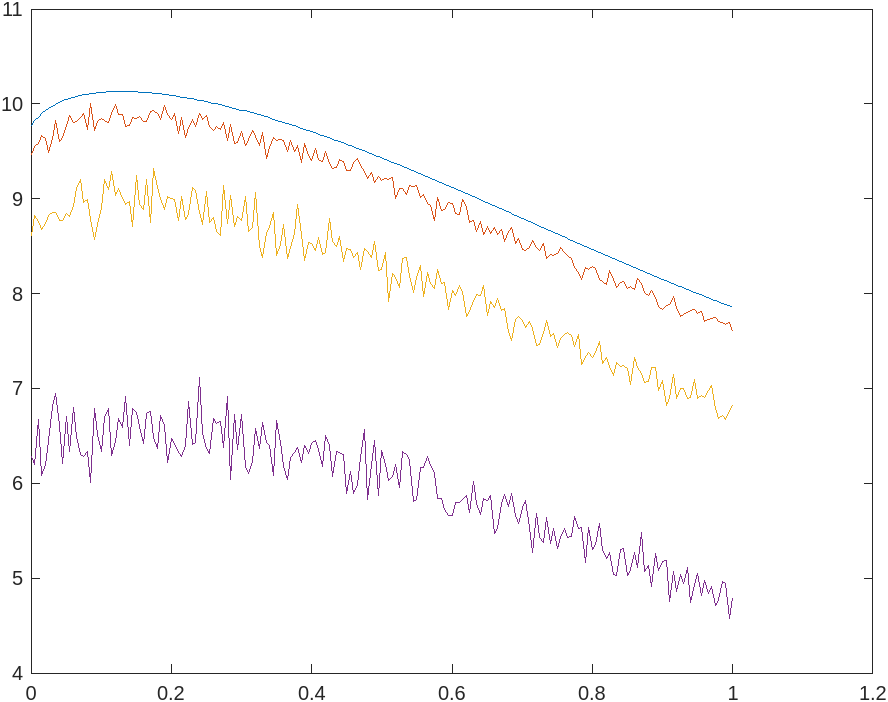}
            \caption{\textbf{Beneficial turning detrimental}}       
    \end{subfigure}
     \begin{subfigure}{0.45\textwidth}
            \centering
            \includegraphics[width=\textwidth]{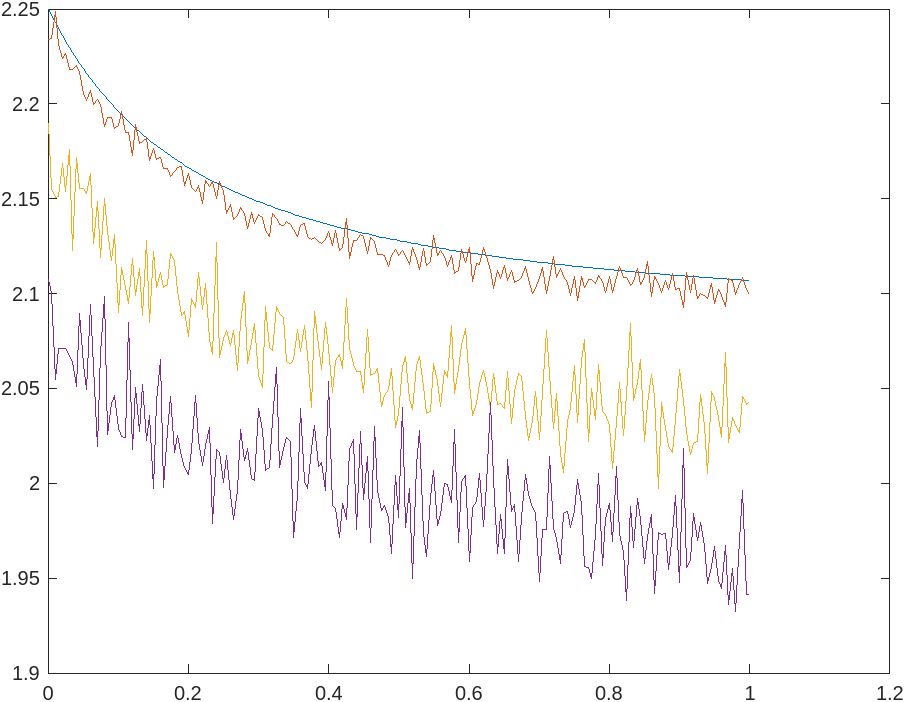}
            \caption{\textbf{Monotonically detrimental}}       
    \end{subfigure}    
    \caption{Total population size versus dispersal rate ($\delta$) in the 2D Beverton-Holt model. We assume the random variables $K_1(t),K_2(t)$ to be uncorrelated: $Cov((K_1(t),K_2(t)))=sI$. The model parameters $r_1,r_2,K_1,K_2$ are the same as in Figure \ref{f:BH6}. The blue curves in all plots depict the deterministic population.  \textbf{A}: the red curve corresponds to $s=5\times 10^{-3}$, the yellow curve to $s=0.02$, and the purple one to $s=0.05$. \textbf{B} the red curve corresponds to $s=5\times 10^{-3}$, the yellow one to $s=0.02$, and the purple one to $s=0.05$ . \textbf{C}: the red curve corresponds to $s=0.05$, the yellow to $s=0.25$, and the purple $s=0.85$ . \textbf{D}: the red curve corresponds to $s=5\times 10^{-3}$, the yellow to $s=0.05$, and the purple $s=0.1$. We observe that in all four cases, adding noise decreases total population while preserving the qualitative relationship between total population and dispersal rate ($\delta$). }
\label{f:BH7}
\end{figure}

\subsection{Two-patch Hassell Model with dispersal}

We also looked at the Hassell model with dispersal.
\begin{equation}
\begin{aligned}
N_1(t+1)&=\frac{(1-\delta)\alpha_1N_1(t)}{(1+K_1(t)N_1(t))^c} +  \frac{\delta \alpha_2N_2(t)}{(1+K_2(t)N_2(t))^c}\\
N_2(t+1)&=\frac{(1-\delta)\alpha_2N_2(t)}{(1+K_2(t)N_2(t))^c} + \frac{\delta \alpha_1N_1(t)}{(1+K_1(t)N_1(t))^c}.
\end{aligned}
\end{equation}

The numerical experiments can be seen in Figure \ref{f:BH8}. We see that, as in the setting without dispersal from \cite{HS25a}, noise seems to often be beneficial in the Hassell model, while it usually is detrimental in the Beverton-Holt model.

\begin{figure}

    \begin{subfigure}{0.45\textwidth}            
            \includegraphics[width=\textwidth]{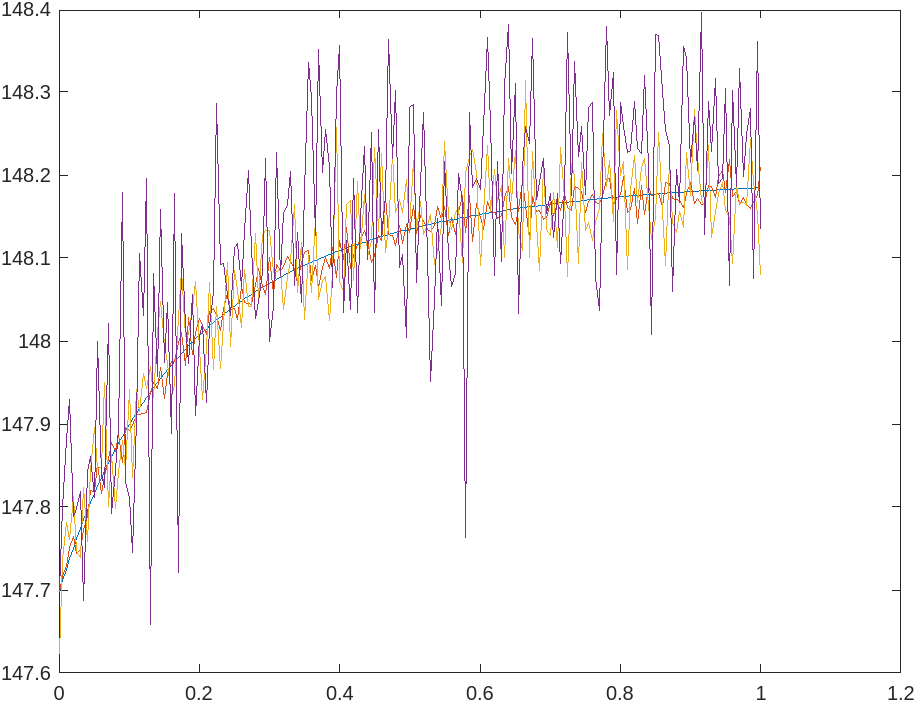}
            \caption{\textbf{Monotonically beneficial}}
           
    \end{subfigure}
    \begin{subfigure}{0.45\textwidth}
            \centering
            \includegraphics[width=\textwidth]{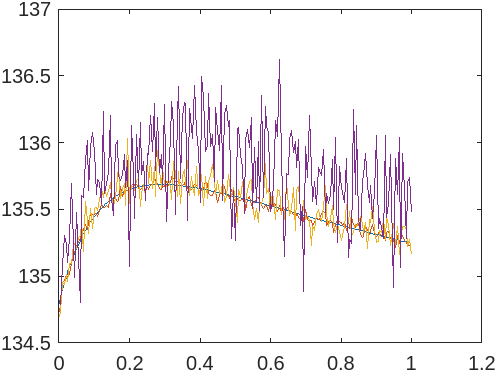}
            \caption{\textbf{Unimodally benefical}}       
    \end{subfigure}
      \begin{subfigure}{0.45\textwidth}
            \centering
            \includegraphics[width=\textwidth]{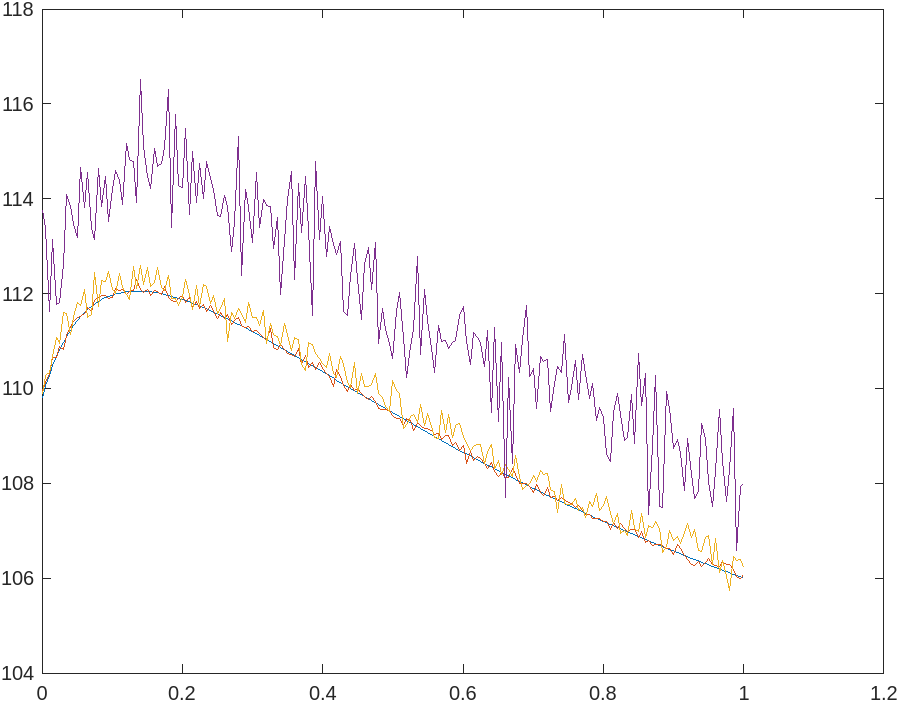}
            \caption{\textbf{Beneficial turning detrimental}}       
    \end{subfigure}
     \begin{subfigure}{0.45\textwidth}
            \centering
            \includegraphics[width=\textwidth]{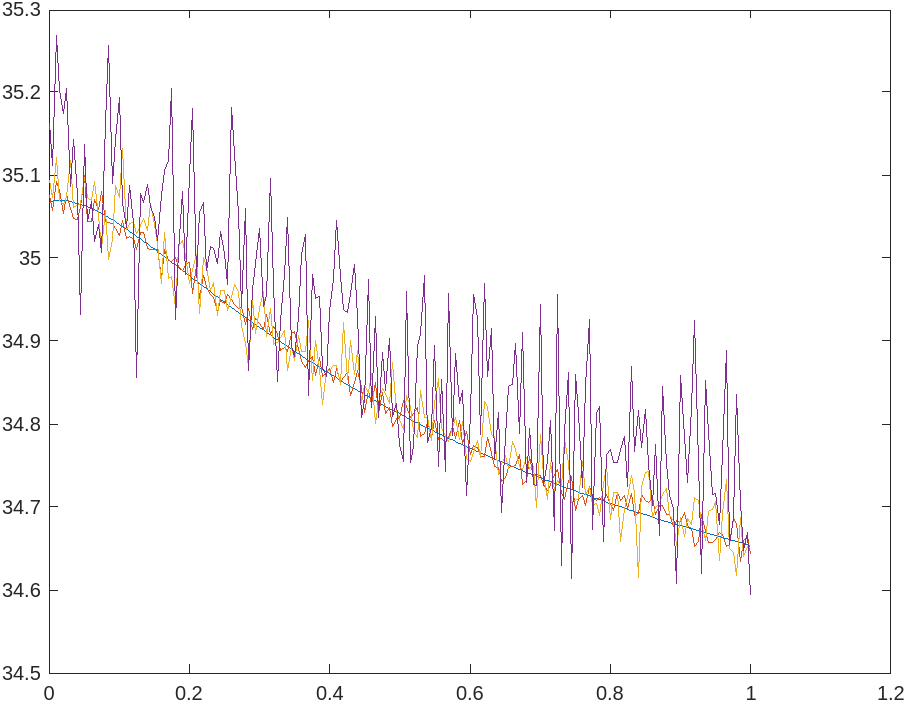}
            \caption{\textbf{Monotonically detrimental}}       
    \end{subfigure}    
    \caption{Total population size versus dispersal rate ($\delta$) in the 2D Hassell model. We assume the random variables $K_1(t),K_2(t)$ to be uncorrelated: $Cov((K_1(t),K_2(t)))=sI$. The blue curves in all plots depict the deterministic population. \textbf{A}: $\alpha_1=2.25$, $K_1=0.025$ and $\alpha_2=1.862$ $K_2=0.0185$. The red curve corresponds to $s=1\times 10^{-5}$, yellow to $s=1\times 10^{-4}$, and purple $s=5\times 10^{-4}$. \textbf{B}:  $\alpha_1=2.25$, $K_1=0.025$ and $\alpha_2=1.75$, $K_2=0.0195$. The red curve corresponds to $s=5\times 10^{-5}$, yellow to $s=5\times 10^{-4}$, and purple $s=5\times 10^{-3}$. \textbf{C}:  $\alpha_1=2.25$, $K_1=0.025$ and $\alpha_2=1.5$, $K_2=0.0225$. the red curve corresponds to $s=5\times 10^{-4}$, yellow to $s=5\times 10^{-3}$, and purple $s=0.05$ .  \textbf{D}: $\alpha_1=1.783$, $K_1=0.0588$ and $\alpha_2=1.838$, $K_2=0.0725$. The red curve corresponds to $s=1\times 10^{-4}$, yellow to $s=5\times 10^{-4}$, and purple $s=5\times 10^{-3}$. We observe that adding small noise does not affect the total population by much, however if the noise is significant we see an increase in total population in panels \textbf{B}, \textbf{C} and \textbf{D}.}
\label{f:BH8}
\end{figure}

\section{Discussion and future work}\label{s:discussion}

In this paper we explored how dispersal and stochasticity impact a population that is spread throughout $n$ patches. Using the results of \cite{BS09} it is possible to have a general result for persistence and extinction when the dynamics is stochastic and the dispersal rates are themselves random. Since the patches are connected by dispersal so that individuals can move between any two patches, the persistence or extinction of all the patches depends on the sign of the Lyapunov exponent, or the metapopulation growth rate, $M$. We show that the assumptions required for the persistence and extinction results are satisfied when the patch dynamics is modeled by Beverton-Holt or Hassell functional responses where the parameters can fluctuate according to i.i.d environments. This allows us to be able to compare the expected population size at stationarity to the population size at the equilibrium of the deterministic dynamics where the stochastic parameters are replaced by their averages. 

Information about the total population size at stationarity is usually hard to quantify analytically because even in the deterministic model, one can only find the unique fixed point explicitly in very special settings, like two-patch models with Beverton-Holt functional responses \cite{grumbach2023effect}. Nevertheless, we can get some analytic expressions in some interesting limiting cases. More specifically, we look at slow and fast dispersal in the Beverton-Holt setting. In these cases both $M$ and the total population size at stationarity can be found explicitly or approximated analytically. In the small dispersal limit it is enough to have one persistent patch - this can save all the other patches, which in the absence of dispersal can be sink patches, from extinction. If only the carrying capacity fluctuates we show in the small dispersal setting that environmental stochasticity is always detrimental as it decreases the total population size. In the high dispersal limit it is possible for all the patches to be sinks and yet have persistence of the coupled system, showing once again that dispersal can save the populations from extinction and increase the total population size. In the Hassell slow dispersal setting we show that noise always increases the total population size, in contrast to what happens in the Beverton-Holt setting. This happens because, as was shown previously in \cite{HS25a}, it is important what we take for the baseline no-noise model - do we keep the carrying capacity constant at its average, or the growth rate constant at its average, or the inverse of the carrying capacity constant at its average?   

Noise is not always due to i.i.d environments. To explore this setting, we look at what happens when the environment changes according to a discrete-time finite-state Markov chain. In this setting we get analytical approximations for the expectation of the total population size at stationarity when the switching of the Markov chain is slow, i.e. the process spends a long time in a state before jumping to another state, as well as when the switching is fast, i.e. the process basically switches states at every time step. In the small dispersal limit we can show that fast switching environments always lead to a higher total population at stationarity than slow switching environments.

Making use of small-noise results, if we also consider that the dispersal rates are small, we are able to obtain explicit results in the Beverton-Holt model when there are $n=2$ patches. We uncovered some interesting behavior: 1) If one only varies the growth rates in the two patches and keeps the dispersal rate $0$ we get that the total population size decreases because of the environmental fluctuations. 2) If the dispersal is fluctuating and $\delta(\xi)=\rho \xi_1$ while the carrying capacities are random and are fluctuating but are perfectly correlated i.e. $K_1 =\bar K_1 + \rho \xi_2, K_2 =\bar K_2 + \rho \xi_2$ then if one patch has a higher carrying capacity $K_1>K_2$ and a lower growth rate $r_2>r_1$ this will increase the total population size. If a patch is has a higher carrying capacity and a higher growth rate then the total population size decreases. The second order terms are more complicated and involve the correlations $\E \xi_1 \xi_2>0$. 3) If the dispersal is fluctuating and $\delta(\xi)=\rho \xi_1$ while the carrying capacities are random and are fluctuating but are driven by different noise terms i.e. $K_1 =\bar K_1 + \rho \xi_2, K_2 =\bar K_2 + \rho \xi_3$ we get the same result as in case 2) but the second order terms are significantly more complex and depend on $\E \xi_1 \xi_2$ and $\E \xi_1 \xi_3$.

\textbf{Future work.} Our current results are only for discrete time models whereas many important biological models are in continuous time. It will be interesting to see how the population size changes due to environmental fluctuations for various types of models, like stochastic differential equations \citep{EHS15}. It is not clear how to generalize the small noise result of \cite{cuello2019persistence} to continuous time dynamics. One possible approach would be to prove small-noise estimates for functionals of the invariant measure around a stable equilibrium.

\textbf{Acknowledgments:} The authors thank Sebastian Schreiber and Peter Chesson, as well as the two anonymous referees, for helpful suggestions which led to an improved manuscript. A. Hening acknowledges generous
support from the NSF through the CAREER grant DMS-2339000. 

\textbf{Conflict of interest.} The authors have no conflict of interest to declare.

\appendix

 \section{Small noise}\label{a:small_noise}

The deterministic dynamics is given by 
\[
\bar N_i(t+1) = F_i(\bar \NN(t), \bar e),
\]
which we will write in vector form as
\[
\bar N(t+1) = F(\bar \NN(t), \bar e).
\]

The dynamics is made stochastic by letting $(e(t))_{t\geq 0}$ be iid random variables taking values in a Polish space.
Assume the following.
\begin{asp}\label{a:Cuello2}
One has:
\begin{itemize}
    \item [(\textbf{B1})] $e(t)=\bar e +\rho\xi(t)$ for each $t$ where $(e(t))_{t\geq 0}$ are iid random variables which take values in a convex compact subset of $\R^\ell$ and $\rho\geq 0$.
    \item [(\textbf{B2})] If $\bar\NN_0 \in \Se$ is an equilibrium for $G(x)=F(x,\bar e)$, then $\bar\NN_0$ is hyperbolic. 
\end{itemize}
\end{asp}
The stochastic dynamics is
\begin{equation}\label{e:gen3}
N_i(t+1) = F_i(\NN(t), e(t)).
\end{equation}

We will adapt and generalize small-noise results by \cite{stenflo1998ergodic} and \cite{cuello2019persistence} in order to find approximations of the expected population size at stationarity.

The first result tells us that if we start close enough to a locally stable equilibrium then the distribution will converge to a unique invariant probability measure which lives in a neighborhood of the equilibrium.
\begin{thm} (Theorem 3.3.1 from \cite{cuello2019persistence})\label{sfp_stoch2}
Consider the process given by $\eqref{e:gen3}$. If Assumption \ref{a:Cuello2} above holds and the equilibrium $\bar\NN_0$ is locally stable, then there exists $\delta>0$ and $\rho>0$ such that  $0<\rho<\delta$ and for all $\NN(0)\in \Se$ with $\|\NN(0)-\bar\NN_0\|<\delta$, we have that the distribution of $\NN(t)$ converges as $t\to \infty$ to a unique invariant probability measure $\mu$ such that $\bar\NN_0\in supp(\mu)$.
\end{thm}

The next result gives the explicit approximations of the expectation of the population at stationarity. 

\begin{thm}\label{t:small_noise_A}
Suppose the Assumptions of Theorem \ref{t:dispersal_det} hold for the system \eqref{e:dispersal_det2}, that $\lambda(D \Lambda(0))>0$ and that Assumption \ref{a:Cuello2} holds. Let $\bar\NN_0$ be the global equilibrium of the deterministic system \eqref{e:dispersal_det2}. There exists $\delta>0$ such that if $|\NN(0)-\bar\NN_0|<\delta$ we have $|\NN(t)-\bar\NN_0|<\delta$ for all $t$ large enough and the occupation measures of $\NN(t)$ converge weakly to an invariant probability measure $\mu$ supported in $B_\delta(\bar\NN_0)$.

Let $D_x F$ be the Jacobian at $\bar\NN_0$ with respect to the population i.e. the $i-j$th entry of $D_x F$ is $\frac{\partial F_i}{\partial x_j}(\bar\NN_0)$. Similarly $D_e F$ will be the matrix with entries $\diag(\frac{\partial F_i}{\partial e_i}(\bar\NN_0))$.
Let $\NN^\rho$ be a vector with distribution $\mu$. We have
 \begin{equation}
\begin{split}
\E(N)=&\bar N_0 + \rho BD_{\xi}F(\bar N_0,0)(\E(\xi))+ \frac{1}{2}B\left(\sum_{i,j=1}^n \frac{\partial^2 F}{\partial N_i\partial N_j}(\bar N_0,0)[\text{Cov}(\NN)]_{ij}\right)\\
+&\frac{\rho^2}{2}B\left(\frac{\partial^2 F}{\partial \xi_i\partial \xi_j}[\text{Cov}(\xi)]_{ij}\right)+ \frac{\rho^2}{2}B\left(\sum_{i,j=1}^n\frac{\partial^2 F}{\partial \xi_i\partial \xi_j}(\bar N_0,0)\E(\xi)_i\E(\xi)_j\right)\\
+&\frac{\rho^2}{2}B\left(\sum_{i,j=1}^n \frac{\partial^2 F}{\partial N_i\partial N_j}(x^*,0)(BD_{\xi}F(\E\xi))_i(B D_{\xi}F_I(\E\xi))_j \right)\\
+&\rho^2B\sum_{i=1}^n\sum_{j=1}^n\frac{\partial^2 F}{\partial N_i\partial \xi_j}(\bar N_0,0)(\E(\xi))_j(BD_{\xi}F(\E\xi))_i+ O(\rho^3)    
\end{split}
\end{equation}
where $B$ denotes the matrix $\left(I-D_NF(\bar N_0,0)\right)^{-1}$ and $ \text{Cov}(N_{ij})$ are the entries of the matrix defined by
\[
vec(\text{Cov}(N)) = (I-D_NF \otimes D_NF)^{-1} (D_{\xi}F \otimes D_{\xi}F) vec(Cov(\xi)).
\]
Moreover, if one can show that \eqref{e:gen3} has an invariant distribution $\bar \mu \neq \delta_0$ and for all $\NN(0)\neq 0$ the empirical measures of $\NN(t)$ converge to $\bar \mu$ then $\bar \mu = \mu$ and the small approximation results from above work globally for any initial condition $\NN(0)\neq 0$. 
\end{thm}

\begin{rmk}
We note that we can get from Theorems \ref{t:BS09}, \ref{t:stoc_m} and \ref{t:stoc_m_h} natural conditions under which the process from \eqref{e:gen3} converges to a unique invariant probability measure, both in the sense of distribution and that of empirical measures. 

It is important to see that for small noise, it is impossible to get $M<0$. We can argue by contradiction as follows. If $M<0$ then by Theorem \ref{t:stoc_m} one has $\NN(t)\to 0$ as $t\to\infty$ which contradicts that by Theorem \ref{sfp_stoch2} if we start with $\NN(0)$ close enough to $\bar\NN_0$ then $\NN(t)$ converges in distribution as $t\to\infty$ to a measure $\mu$ supported in a small ball around $\bar \NN_0$.

As a result, for small noise we will have $M>0$ -- intuitively, a persistent deterministic system will still be persistent if one adds small noise to it.
\end{rmk}
\begin{rmk}
One can show under natural assumptions that the Lyapunov exponent $M$ depends continuously on the model parameters -- see Proposition 3 from \cite{BS09}.     
\end{rmk}

\begin{proof}
We know from Proposition 3.6.2 in \cite{cuello2019persistence} that there exists a $\delta>0$, such that if $\|\NN(0)-\bar\NN_0\|<\delta$, and $\rho\in (0,\delta)$ then for all $t\in \N$, we have $\|\NN(t)-\bar\NN_0\|<\delta$. By the statement of Theorem \ref{t:small_noise_A} we know that $\bar N$ is a global equilibrium for the deterministic dynamical system $\bar N_i(t+1)=F_i(\bar \NN(t),\bar e)$ which implies  $\bar \NN(t)\to \bar \NN_0$ as $t\to\infty$. It follows from Theorem 3.3.1 and Proposition 3.6.2 in \cite{cuello2019persistence}, that if $\|\NN(0)-\bar\NN_0\|<\delta$ the occupational measures of $\NN(t)$ converge to $\mu$ and $supp(\mu)\subset B(\bar \NN_0,\delta)$. 

We cannot use Theorem 3.3.2 from \cite{cuello2019persistence} for the small noise expansion since in \cite{cuello2019persistence} the assumption is that $\E \xi=0.$ We modify \cite{cuello2019persistence} and argue instead as follows.

Let $(Y_1(t))_{t\geq 0}$ be the first order approximation of $N(t)$ given by
\begin{equation}
Y_1(t+1)=\bar N_0 + D_NF(Y_2(t)-\bar N_0) + \rho D_{\xi}F(\xi(t+1))    
\end{equation}
Taking expectation on both sides, noting that due to linearity one can switch the expectation with the matrix multiplication, we see that
\begin{equation}
\begin{split}
\E(Y_1(t+1))=\bar N_0 + D_NF(\E(Y_1(t+1))-\bar N_0) + \rho D_{\xi}F(\E\xi).    
\end{split}    
\end{equation}
By arguing as in Proposition 3.6.4. from \cite{cuello2019persistence} we get that $Y_1(t)\to\mu_1$ in distribution to some random variable $Y_1$ which has distribution $\mu_1$. Since $Y_1$ is bounded, note that $$\E(Y_1(t))\to \E(Y_1)$$ as $t\to\infty$. Assuming that $\left(I-D_NF(\bar N,0)\right)^{-1}$ exists we get that
\[
\E(Y_1)= \bar N_0 + \rho (I-D_NF)^{-1} D_{\xi}F(\E\xi).
\]
Since we have $\E(N)=E(Y_1)+O(\rho^2)$ we get 

\begin{equation}
\E(N)=  \bar N_0 + \rho (I-D_NF)^{-1} D_{\xi}F(\E\xi) + O(\rho^2).  
\end{equation}
Now to get an expression for $Cov(Y_2)$. We define $Y_2(t)$ in the following manner.

\begin{equation}
\begin{split}
Y_2(t+1)=&\bar N_0 + D_NF(Y_2(t)-)+ \rho D_{\xi}F\xi(t+1)+ \frac{1}{2}D_{NN}F(Y_2(t)-\bar N_0,Y_2(t)-\bar N_0)\\
+&\frac{\rho^2}{2}D_{\xi\xi}F(\xi(t+1),\xi(t+1))+\rho D_{N\xi}F(Y_2(t)-\bar N_0,\xi(t+1))
\end{split}  
\end{equation}
Taking expectation on both sides we get
\begin{equation}
\begin{split}
\mathbb{E}(Y_2(t+1))&=\bar N_0 + D_NF_I(\E Y_2(t)-\bar N_0)+ \rho D_{\xi}F(\E\xi)+ \E\left(\frac{1}{2}D_{NN}F_I(Y_2(t)-\bar N_0,Y_2(t)-\bar N_0)\right)\\
&~+\frac{\rho^2}{2}\E\left(D_{\xi\xi}F(\xi,\xi)\right)+\rho \E\left(D_{N\xi}F(Y_2(t)-\bar N_0,\xi)\right)\\
&=\bar N_0 + D_xF(\E Y_2(t)-\bar N_0)+ \rho D_{\xi}F(\E\xi)+ \frac{1}{2}\sum_{j,k=1}^{n} \frac{\partial^2F}{\partial N_j\partial N_k}\E\left((Y_2(t)-\bar N_0)_j(Y_2(t)-\bar N_0)_k \right)\\
&~+\frac{\rho^2}{2}\sum_{j,k=1}^m \frac{\partial^2F}{\partial \xi_j\partial \xi_k}\E(\xi_j(t+1)\xi_k(t+1)) + \rho \sum_{j=1}^{n}\sum_{k=1}^m \frac{\partial^2 F}{\partial N_j\partial \xi_k}\E( (Y_2(t)-\bar N_0)_j\xi_k(t+1)).
\end{split}    
\end{equation}
Using the same idea of as above we get
\begin{equation}
\begin{split}
\E(Y_2)=&\bar N_0 + D_NF(\E Y_2-\bar N_0)+ \rho D_{\xi}F(\E\xi)+ \frac{1}{2}\sum_{j,k=1}^{n} \frac{\partial^2F}{\partial N_j\partial N_k}\E\left((Y_2-\bar N_0)_j(Y_2-\bar N_0)_k \right)\\
+&\frac{\rho^2}{2}\sum_{j,k=1}^m \frac{\partial^2F}{\partial \xi_j\partial \xi_k}\E(\xi_j\xi_k) + \rho \sum_{j=1}^{n}\sum_{k=1}^m \frac{\partial^2 F}{\partial N_j\partial \xi_k}\E( (Y_2-\bar N_0)_j\xi_k)   \end{split}.    
\end{equation}
We see that
\begin{equation}
\begin{split}
\E(Y_2)\E(Y_2)'=&\bar N_0\E(Y_2)'+ D_NF(\E Y_2-\bar N_0)\bar N_0'+ D_NF(\E Y_2-\bar N_0)D_NF(\E Y_2-\bar N_0)'\\
+&\rho D_NF(\E Y_2-\bar N_0)D_\xi F(\E(\xi))'+\rho D_\xi F(\xi)\bar N_0'+\rho D_\xi F(\E(\xi))D_{\xi}F(\E Y_2-\bar N_0)'\\    
+&\rho^2D_{\xi}F(\E(\xi))\left(D_{\xi}F(\E(\xi) )\right)' + \frac{1}{2}\E(D_{NN}F(Y_2-\bar N_0,Y_2-\bar N_0))N_0'\\
+&\frac{\rho^2}{2}\E(D_{\xi\xi}F(\xi,\xi))\bar N_0'+ \rho \E(D_{N\xi}F(Y_2-\bar N_0,\xi))\bar N_0' + O(\rho^3).
\end{split}    
\end{equation}
Expanding $Y_2(t+1)Y_2(t+1)'$ and taking expectations yields
\begin{equation}
\begin{split}
\E(Y_2(t+1)Y_2(t+1)') =& \bar N_0\E(Y_2(t+1))'+ D_NF(\E(Y_2(t)-\bar N_0))\bar N_0'\\
+& \E\left( D_NF(Y_2(t)-\bar N_0)D_NF(Y_2-\bar N_0)'\right)\\
+& \rho\E\left(D_NF(Y_2(t)-\bar N_0)D_{\xi}F\xi(t+1)' \right)+\rho D_{\xi}F\E(\xi)\bar N_0'\\ 
+& \rho \E\left(D_{\xi}F\xi(t+1)D_NF(Y_2(t)-\bar N_0)'\right)\\
+& \rho^2\left(D_{\xi}F\xi(t+1)D_{\xi}F\xi(t+1)' \right)+ \frac{1}{2}\E(D_{NN}F(Y_2(t)-\bar N_0,Y_2(t)-\bar N_0))\bar N_0'\\
+&\frac{\rho^2}{2}\E\left(D_{\xi\xi}F(\xi(t+1),\xi(t+1))\right)\bar N_0\\
+& \rho\E\left(D_{N\xi}F(Y_2(t)-\bar N_0,\xi(t+1))\right)\bar N_0'
+ O(\rho^3). 
\end{split}
\end{equation}
Since $(\xi(t))_{t\in\mathbb{N}_0}$ is iid we have 
\begin{equation}
  \begin{split}
  \E\left( D_NF(Y_2(t)-\bar N_0)D_\xi F\xi(t+1)'\right)=&\E\left(D_NF(Y_2(t)-\bar N_0)\right)\E(D_\xi F\xi(t+1)')\\
  =&D_NF\E(Y_2(t)-\bar N_0)[D_\xi F][\E(\xi(t+1))]'.    
  \end{split}  
\end{equation}
Because $Y_2(t)\to Y_2$ in distribution by Proposition 3.6.4. from \cite{cuello2019persistence}, and for $t$ large enough $\|Y(t)-\bar N_0\|\leq C\rho$, we see that $Y_2(t)Y_2'(t)\to Y_2Y_2'$ in distribution. As a result
\begin{equation}
\begin{split}
\E(Y_2Y_2')=&\bar N_0\E(Y_2)+ D_NF(\E(Y_2-\bar N_0))\bar N_0'+ \E\left( D_NF(Y_2-\bar N_0)D_NF(Y_2-\bar N_0)'\right)\\
+& \rho\E\left(D_NF(Y_2-\bar N_0)D_{\xi}F\xi' \right)+\rho D_\xi F\E(\xi)\bar N_0 + \rho \E\left(D_eF\xi D_NF(Y_2-\bar N_0)'\right)\\
+& \rho^2\left(D_\xi F\xi D_\xi F\xi' \right)+ \frac{1}{2}\E(D_{NN}F(Y_2-\bar N_0,Y_2-\bar N_0))\bar N_0\\
+&\frac{\rho^2}{2}\E\left(D_{\xi\xi}F(\xi,\xi\right)\bar N_0'+ \rho\E\left(D_{N\xi}F(Y_2-\bar N_0,\xi)\right)\bar N_0+ O(\rho^3)     
\end{split}    
\end{equation}
Hence we have the expression for $\text{Cov}(Y_2)$ 
\begin{equation}\label{Cov_1}
\begin{split}
\text{Cov}(Y_2)=\E(Y_2Y_2')-\E(Y_2)\E(Y_2')=&   \E\left( D_NF(Y_2-\bar N_0)D_NF(Y_2-\bar N_0)'\right)\\
-& D_NF(\E(Y_2)-\bar N_0)D_NF(\E(Y_2)-\bar N_0)'
+\rho^2\left(\E(D_{\xi}F(\xi)D_{\xi}F(\xi)')\right)\\
-&\rho^2\left(D_{\xi}F\E(\xi)D_{\xi}F\E(\xi)'\right)+ O(\rho^3).   
\end{split}    
\end{equation}
We simplify \eqref{Cov_1} to get
\begin{equation}
\begin{split}
\E(Y_2Y_2')-\E(Y_2)\E(Y_2')=& \E\left( D_NF(Y_2-\E(Y_2))D_NF(Y_2-\E(Y_2))'\right)\\
+&\rho^2\E\left(D_{\xi}F(\xi-\E(\xi))D_{\xi}F(\xi-\E(\xi)'\right)+O(\rho^3).
\end{split}    
\end{equation}
This implies
\begin{equation}
\begin{split}
\text{Cov}(Y_2)=& D_NF\text{Cov}(Y_2)D_NF +\rho^2D_\xi F\text{Cov}(\xi)D_\xi F+O(\rho^3).
\end{split}
\end{equation}
As a result 
\begin{equation}
\begin{split}
Vec(Cov(Y_2))=&[D_NF'\otimes D_NF]Vec(Cov(Y_2))+\rho^2[D_{\xi}F\otimes D_{\xi}F]Vec(Cov(\xi))+O(\rho^3).
\end{split}
\end{equation}
Hence we get 
\begin{equation}\label{eqCOV}
\begin{split}
Vec(Cov(N))=\rho^2[I-D_NF_I'\otimes D_NF]^{-1}[D_{\xi}F\otimes D_{\xi}F]Vec(Cov(\xi)) + O(\rho^3).
\end{split}    
\end{equation}
Recall that for the expectation we have 
\begin{equation}\label{eqexp}
\E(N)=  \bar N_0 + \rho (I-D_NF)^{-1} D_{\xi}F_I(\E\xi) + O(\rho^2).    
\end{equation}
By taking expectations in
\begin{equation}
\begin{split}
N(t+1)=&\bar N_0 +  D_NF(\bar N_0,0)(X(t)-\bar N_0) + \rho D_{\xi}F(\bar N_0,0)(\xi(t+1))\\
+& \frac{1}{2}(X(t)-\bar N_0)'D_{NN}F(\bar N_0,0)(X(t)-\bar N_0)\\
+&\frac{\rho^2}{2}(\xi(t+1))'D_{\xi\xi}F(\bar N_0,0)(\xi(t+1))\\
+& \rho\sum_{i=1}^n\sum_{j=1}^n\frac{\partial^2 F}{\partial N_i\partial \xi_j}(\bar N_0,0)(\xi(t+1))_j(X(t)-\bar N_0)_i + O(\rho^3)
\end{split}   
\end{equation}
we get
\begin{equation}
\begin{split}
\mathbb{E}(N(t+1))=&\bar N_0 +  D_NF(\bar N_0,0)(\E(N(t))-\bar N_0) + \rho D_{\xi}F(\bar N_0,0)(\E(\xi))\\
+& \frac{1}{2}\E\left[(N(t)-\bar N_0)'D_{NN}F(\bar N_0,0)(N(t)-\bar N_0)\right]\\
+&\frac{\rho^2}{2}\E\left[(\xi)'D_{\xi\xi}F(\bar N_0,0)(\xi)\right]\\
+& \rho\sum_{i=1}^n\sum_{j=1}^n\frac{\partial^2 F}{\partial N_i\partial \xi_j}(\bar N_0,0)(\E(\xi))_j(\E(N(t))-\bar N_0)_i + O(\rho^3).    
\end{split}    
\end{equation}
Using the convergence in distribution yields
\begin{equation}
\begin{split}
\E(N)=&\bar N_0 +  D_NF(\bar N_0,0)(\E(N)-\bar N_0) + \rho D_{\xi}F(\bar N_0,0)(\E(\xi))\\
+& \frac{1}{2}\E\left((N-\bar N_0)'D_{NN}F(\bar N_0,0)(N-\bar N_0)\right)\\ 
+&\frac{\rho^2}{2}\E\left((\xi)'D_{\xi\xi}F(\bar N_0,0)(\xi)\right)\\
+& \rho\sum_{i=1}^n\sum_{j=1}^n\frac{\partial^2 F}{\partial N_i\partial \xi_j}(\bar N_0,0)(\E(\xi))_j(\E(X)-\bar N_0)_i + O(\rho^3)
\end{split}    
\end{equation}
The matrix $D_NF(\bar N_0,0)$ is invertible so
\begin{equation}\label{expression_exp}
\begin{split}
\E(N)=& \bar N_0 + \rho(I-D_NF)^{-1}D_{\xi}F(\bar N_0,0)(\E(\xi))\\
+&\frac{1}{2}\left(I-D_{N}F(\bar N_0,0)\right)^{-1}\E\left((N-\bar N_0)'D_{NN}F(\bar N_0,0)(N-\bar N_0)\right)\\
+&\frac{\rho^2}{2}(I-D_NF(\bar N_0,0))^{-1}\E\left((\xi)'D_{\xi\xi}F(\bar N_0,0)(\xi)\right)\\
+&\rho(I-D_NF)^{-1}\sum_{i=1}^n\sum_{j=1}^n\frac{\partial^2 F}{\partial N_i\partial \xi_j}(\bar N_0,0)(\E(\xi))_j(\E(N)-\bar N_0)_i.
\end{split}    
\end{equation}
Using \eqref{eqexp} we can modify the last term of \eqref{expression_exp}
\begin{equation}
\begin{split}
\E(N)=& x^* + \rho(I-D_NF(\bar N_0,0))^{-1}D_\xi F(\bar N_0,0)(\E(\xi))\\
+&\frac{1}{2}\left(I-D_N F(\bar N_0,0)\right)^{-1}\E\left((N-\bar N_0)'D_{NN}F(\bar N_0,0)(N-\bar N_0)\right)\\
+&\frac{\rho^2}{2}(I-D_NF(\bar N_0,0))^{-1}\E\left((\xi)'D_{\xi\xi}F(\bar N_0,0)(\xi)\right)\\
+&\rho^2(I-D_NF(\bar N_0,0))^{-1}\sum_{i=1}^n\sum_{j=1}^n\frac{\partial^2 F}{\partial N_i\partial \xi_j}(\bar N_0,0)(\E(\xi))_j((I-D_NF)^{-1} D_\xi F(\E\xi))_i\\
+&O(\rho^3).
\end{split}    
\end{equation}
We rewrite the expression with the modification that $B$ denotes $(I-D_NF)^{-1}$
and using the fact that $[\text{Cov}(\xi)]_{ij}+\E(\xi_i)\E(\xi_j)=\E(\xi_i\xi_j)$ 
\begin{equation}\label{taylor_fin}
\begin{split}
\E(N)=& \bar N_0 + \rho B D_\xi F(\bar N_0,0)(\E(\xi)) + \frac{1}{2}B\left(\sum_{i,j=1}^n \frac{\partial^2 F}{\partial N_i\partial N_j}(\bar N_0,0)[\text{Cov}(N)]_{ij}\right)\\
+&\frac{\rho^2}{2}B\left(\frac{\partial^2 F}{\partial \xi_i\partial \xi_j}(\bar N_0,0)\E(\xi_i\xi_j)\right)+\frac{\rho^2}{2}B\left(\sum_{i,j=1}^n \frac{\partial^2 F}{\partial N_i\partial N_j}(\bar N_0,0)(B D_{\xi}F(\E\xi))_i(B D_{\xi}F(\E\xi))_j \right)\\
+&\rho^2B\sum_{i=1}^n\sum_{j=1}^n\frac{\partial^2 F}{\partial N_i\partial \xi_j}(\bar N_0,0)(\E(\xi))_j(B D_\xi F(\E\xi))_i+ O(\rho^3).
\end{split}    
\end{equation}
\end{proof}

\section{Verification of Assumptions in the Beverton-Holt model with fluctuating dispersal}\label{a:BH_d}
We want to make use of Theorem 1 from \cite{BS09}. We note that the computations are similar to those done in \cite{BS09} and require minimal modifications. They are added here for completeness. 
\begin{asp}\label{a:BH_d}
Suppose the dynamics is given by

\begin{equation}
    \begin{split}
      \BX(t+1)& =F(\BX(t),\xi(t))= D(t)\Lambda(t) \BX(t):\\
&\left(D(t)\diag\left(  \frac{r_1(t)}{1+\kappa_1(t)X_1(t)},\dots,\frac{r_n(t)}{1+\kappa_n(t)X_n(t)}\right)\right)\BX(t)  
    \end{split}
\end{equation}

where $D(t):=(d_{ij}(t)), r_i(t)$ and $\kappa_i(t)$ can be random. More specifically, the randomness is given by $$\xi(t)=(r_i(t), \kappa_i(t), d_{ij}(t))$$ which is assumed to form an iid sequence in $\R^\ell$ where $\ell=2n+n^2$.
We will assume the following:
\begin{itemize}
    \item $\E\left(\max_j \frac{r_j}{\kappa_j}\right)<\infty$.
    \item $\E\left(\ln^+\left(\max_j \frac{r_j}{\kappa_j}+ \max_j r_j + 2 \max_j (r_j\kappa_j)\right)\right)<\infty$
    \item With probability one: $r_i(1)>0, \kappa_i(1)\geq 0, 0<d_{ij}(1)<1$ and $\sum_{j=1}^nd_{ij}(1)=1$ for $i,j=1,\dots n$. 
\end{itemize}

\end{asp}

\begin{prop} Assumptions A1-A6 from \cite{BS09} hold. The Lyapunov exponent $M$ exists and as a result, if $M>0$, then $\BX(t)\to \BX(\infty)$ in distribution as $t\to\infty$ while if $M<0$ then $\NN(t)\to 0$.
\end{prop}
\begin{proof}

\textbf{Assumption A1} from \cite{BS09} holds because we assume the $\xi(t)$'s form an iid sequence in $\R^\ell$.

We check \textbf{Assumption A2} from \cite{BS09}. We find a proper function $W:\R_+^n\to \R$ and random variables $\gamma,\tau:\Omega\to [0,\infty)$ such that the following hold
\begin{itemize}
    \item [(i)] $W(D(\omega)\Lambda(x,\xi(\omega))x)\leq \gamma(\omega)W(x)+\tau(\omega)$
    \item [(ii)] $\E(\ln(\gamma))<0$
    \item [(iii)] $\E(\ln^{+}(\tau))<\infty$
\end{itemize}
Instead of writing $r_i(t), \kappa_i(t), d_{ij}(t)$ we slightly abuse notation and just write $r_i, \kappa_i, d_{ji}$ below.

Following \cite{BS09} let $W(x) := \sum_{i=1}^n x_i$ and suppose $D=(d_{ij})$ is a random dispersal matrix. Note that 
$$(D\Lambda)_{ij}=\left(D\diag\left( \frac{r_i}{1+\kappa_ix_i}\right)\right)_{ij} = \left(\frac{d_{ij}r_j}{1+\kappa_j x_j}\right).$$ 
Since $x_j\geq 0$ we have
\[
\frac{d_{ij}r_jx_j}{1+\kappa_j x_j} \leq \frac{d_{ij}r_j}{\kappa_j}.
\]

Using this and the fact that $\sum_i d_{ij}=1$ we get the following estimates
\begin{equation}
    \begin{split}
 W\left(D\Lambda x\right) =& W\left(D\diag\left( \frac{r_i}{1+\kappa_ix_i}\right) x\right) \\
 &=   \sum_{i,j} \frac{d_{ij}r_jx_j}{1+\kappa_j x_j}    \\
 &= \sum_j \frac{r_jx_j}{1+\kappa_j x_j}\\
 &\leq \sum_j\frac{r_j}{\kappa_j}\\
 &\leq n \max_j \frac{r_j}{\kappa_j}
    \end{split}
\end{equation}
As such, if we take $\gamma=0$ and $\tau = n \max_j \frac{r_j}{\kappa_j}$ (as in \cite{BS09}) the conditions are satisfied since using Jensen's inequality and Assumption \ref{a:BH_d} we get
\[
\E \ln \tau < \ln n + \ln \E \left(\max_j \frac{r_j}{\kappa_j}\right) < \infty.
\]

\textbf{Assumption A3} from \cite{BS09}: Since the matrix $D$, which can be random, has coefficients $d_{ij}\in(0,1)$ it is clearly primitive. $\Lambda$ is a diagonal matrix with strictly positive entries on the diagonal. The product of a primitive matrix and a diagonal matrix with strictly positive diagonal entries is always primitive which shows that the matrix $D\Lambda(x)$ is primitive. 

We verify \textbf{Assumption A4} from \cite{BS09}, that is $$\E\left(\sup_{\|x\|\leq 1} \ln^{+}(\|F(x)\|+\|DF(x)\|+\|D^2F(x))\|\right)<\infty$$
For this we note that
\[
\sum_i (D\Lambda \bx)_i =\sum_i \sum_k \left(\frac{d_{ik} r_kx_k}{1+\kappa_k x_k}\right)_i \leq n \max_k \frac{r_k}{\kappa_k}
\]
and 
\[
\frac{\partial}{\partial x_l} \sum_k\left(\frac{ d_{ik} r_kx_k}{1+\kappa_k x_k}\right) = \frac{d_{il}r_l}{(1+\kappa_l x_l)^2} \leq d_{il} r_l.
\]
Differentiating again we get for $l\neq j$
\[
\left|\frac{\partial^2}{\partial x_l\partial x_j} \left(\frac{\sum_k d_{ik} r_kx_k}{1+\kappa_k x_k}\right) \right| = 0
\]
and
\[
\left|\frac{\partial^2}{\partial x_l\partial x_l} \sum_k\left(\frac{ d_{ik} r_kx_k}{1+\kappa_k x_k}\right) \right| = \left|-\frac{2d_{il}r_l\kappa_l}{(1+\kappa_k x_l)^3}\right|\leq 2d_{il}r_l\kappa_l
\]
Therefore
\begin{equation*}
    \begin{split}
    \E\left(\sup_{\|x\|\leq 1} \ln^{+}(\|F(x)\|+\|DF(x)\|+\|D^2F(x))\|\right) &\leq \E \left(\ln^+ \left(\max_j \frac{r_j}{\kappa_j}+ \max_j r_j + 2 \max_j r_j\kappa_j\right)\right) \\
    &<\infty.
    \end{split}
\end{equation*}
\\
Verification of \textbf{Assumption A5} from \cite{BS09}: For the matrix $D\Lambda(x,\xi)$ one has
    $$[D\Lambda(x)]_{ij}=\frac{d_{ij}r_j}{1+\kappa_jx_j}$$
which implies
   $$\frac{\partial [D\Lambda(x)]_{ij}}{\partial x_l}= \delta_{lj} \frac{-2d_{ij}r_j\kappa_j}{(1+\kappa_jx_j)^2}$$
clearly we have $\frac{\partial [D\Lambda(x)]_{ij}}{\partial x_l}\leq 0$. Moreover, if $l=j$ we see that $\frac{\partial [D\Lambda(x,\xi)]_{ij}}{\partial x_j}< 0.$
\\
Verification of \textbf{Assumption A6} from \cite{BS09}: Since $$(D\Lambda \bx)_i=\sum_{k=1}^n \frac{d_{ik}r_kx_k}{1+\kappa_kx_k}  $$
we get $$\frac{\partial F_i}{\partial x_j}=\frac{d_{ij}r_j}{(1+r_jx_j)^2}\geq 0.$$

 By the calculations and arguments above we also get that
\[
\E \left| \ln \|A(0,\xi(0)\|\right|<\infty
\]
implying that Assumption \ref{asp:lyapunov} holds. By \cite{ruelle1979analycity, BS09} we get that the Lyapunov exponent $M$ exists. The result now follows from Theorem 1 and Proposition 1 in \cite{BS09}.
\end{proof}
The last statement of Theorem \ref{t:stoc_m} follows from Theorem \ref{t:dispersal_det} because all the assumptions of Theorem \ref{t:dispersal_det} hold, since in particular $\bar D$ is primitive because $0<\bar d_{ij} = \E d_{ij} <1$ from Assumption \ref{a:BH_dd}.

\section{Verification of Assumptions in the Hassell model with fluctuating dispersal} \label{a:hassell}
In this appendix we treat the Hassell functional response. 
\begin{asp}\label{a:H_d}
The dynamics is given by
\begin{equation}
    \begin{split}
\BX(t+1) &=F(\BX(t),\xi(t))= D(t)\Lambda(t) \BX(t):\\
&\left(D(t)\diag\left(  \frac{\alpha_1(t)}{(1+K_1(t)X_1(t))^c},\dots,\frac{\alpha_n(t)}{(1+K_n(t)X_n(t))^c}\right)\right)\BX(t)
    \end{split}
\end{equation}
where $D(t):=(d_{ij}(t)), r_i(t)$ and $\kappa_i(t)$ can be random. More specifically, the randomness is given by $$\xi(t)=(\alpha_i(t), K_i(t), d_{ij}(t))$$ which is assumed to form an iid sequence in $\R^\ell$ where $\ell=2n+n^2$.
We will assume the following:
\begin{itemize}
    \item $c\in (0,1]$
    \item $\E\left(\max_j \frac{\alpha_j}{K_j^c}\right)<1$.
    \item $\E\left(\ln^+\left(\max_j \frac{\alpha_j}{K_j^c}+ \max_j \alpha_j +  \max_j \alpha_jK_j^2\right)\right)<\infty$
    \item With probability one: $\alpha_i(1)>0, K_i(1)> 0, 0<d_{ij}(1)<1$ and $\sum_{j=1}^nd_{ij}(1)=1$ for $i,j=1,\dots n$. 
\end{itemize}
\end{asp}

\textbf{Assumption A1} from \cite{BS09} holds because we assume the $\xi(t)$'s form an iid sequence in $\R^\ell$.

We check \textbf{Assumption A2}. We find a proper function $W:\R_+^n\to \R$ and random variables $\gamma,\tau:\Omega\to [0,\infty)$ such that the following hold
\begin{itemize}
    \item [(i)] $W(D(\omega)\Lambda(x,\xi(\omega))x)\leq \gamma(\omega)W(x)+\tau(\omega)$
    \item [(ii)] $\E(\ln(\gamma))<0$
    \item [(iii)] $\E(\ln^{+}(\tau))<\infty$
\end{itemize}
We write $\alpha_i, K_i,d_{ij}$ instead of $\alpha_i(t),K_i(t),d_{ij}(t)$, which is a slight abuse of notation. 
We proceed similarly to the Beverton-Holt case above. 
$$(D\Lambda)_{ij}=\left(D\diag\left( \frac{\alpha_i}{(1+K_ix_i)^c}\right)\right)_{ij} = \left(\frac{d_{ij}\alpha_j}{(1+K_j x_j)^c}\right).$$
Since we know that $x_j\geq 0$ we have 
\[
\frac{d_{ij}\alpha_jx_j}{(1+K_jx_j)^c}\leq \frac{d_{ij}\alpha_jx_j}{K_j^cx_j^c}=\frac{d_{ij}\alpha_j}{K_j^c}x_j^{1-c}.
\]
We can see that 
\[
D\Lambda x=\left(\sum_i^n\frac{d_{1i}\alpha_ix_i}{(1+K_ix_i)^c}, \dots, \sum_i^n \frac{d_{ni}\alpha_ix_i}{(1+K_ix_i)^c}\right)
\]
We consider $W(x)=\sum_i^n x_i$, and hence we get the estimate
\[
 W(D\Lambda x)=\sum_{i=1}^n\sum_{j=1}^n\frac{d_{ij}\alpha_jx_j}{(1+K_jx_j)^c}\leq \sum_{i=1}^n\sum_{j=1}^n\frac{d_{ij}\alpha_j}{K_j^c}x_j^{1-c}   
\]     
      
Now observe that if $x_j\geq 1$, then $x_j^{1-c}\leq x_j$, else we use $x_j^{1-c}\leq 1$, as a result we get
\[
\sum_{i=1}^n\sum_{j=1}^n\frac{d_{ij}\alpha_j}{K_j^c}x_j^{1-c}\leq \sum_{i=1}^n\sum_{x_j<1}\frac{d_{ij}\alpha_j}{K_j^c}+\sum_{i=1}^n\sum_{x_j\geq 1}\frac{d_{ij}\alpha_j}{K_j^c}x_j
\]
so now using the property that $\sum_i d_{ij}=1$, we can get the following estimate
\[
W(D\Lambda x) = \sum_{j=1}^n\frac{\alpha_j}{K_j^c}+  \left(\max_j\frac{\alpha_j}{K_j^c}\right)W(x)
\]
consider $\gamma=\left(\max_j\frac{\alpha_j}{K_j^c}\right)$ and $\tau=\sum_{j=1}^n\frac{\alpha_j}{K_j^c}$
and we have the desired inequality 
\[
W(D\Lambda x)\leq \gamma W(x)+\tau.
\]
Note that due to the Assumptions on the coefficients we also have $\E(\ln(\gamma))<0$ and $\E(\ln^{+}(\tau))<\infty$.

\textbf{Assumption A3: } The matrix $D$ is primitive, and since $\Lambda$ is a diagonal matrix with strictly positive entries, their product is also primitive. 

\textbf{Assumption A4: }: Under the condition $\|x\|\leq 1$ we have
\[
\sum_i (D\Lambda x)_i =\sum_i\sum_j\left( \frac{d_{ij}\alpha_jx_j}{(1+K_jx_j)^c}\right)\leq  \max_j\frac{\alpha_j}{K_j^c}x_j^{1-c}\leq \max_j \frac{\alpha_j}{K_j^c}
\]
\[
\frac{\partial F_i}{\partial x_j}=\frac{d_{ij}\alpha_j+ (1-c)K_j\alpha_jx_j}{(1+K_jx_j)^{c+1}}\leq d_{ij}\alpha_j+(1-c)\alpha_jK_jx_j
\]
\[
\sum_i\sum_j \frac{\partial F_i}{\partial x_j} \leq n \max_j \alpha_j + n(1-c)\max_j \alpha_jK_j^{1-c}.
\]
Since we have that $\frac{\partial^2F_i}{\partial x_k\partial x_j}=0$ for $j\neq k$, 
\[
\left| \frac{\partial^2F_i}{\partial x_j^2}\right|= \left|\frac{K_j\alpha_j((1-c)-(1+c)d_{ij}-c(1-c)K_jx_j)}{(1+K_jx_j)^{c+2}}\right |\leq K_j\alpha_j(1-c) + (1+c)K_j\alpha_jd_{ij}+ c(1-c)K_j^2\alpha_jx_j
\]
using the above inequality and keeping in mind that $\|x\|\leq 1 $ we get
\[
\sum_i\sum_j \frac{\partial^2 F_i}{\partial x_j^2}\leq 2n\max_j K_j\alpha_j + nc(1-c)\max_j K_j^2\alpha_j.
\]
\\
This implies
\[
\sup_{\|x\|\leq 1}\E(\ln^{+}(\|F(x)\| +\|DF(x)\|+ \|D^2F(x)\|)) <\infty.
\]

\textbf{Assumption A5: } We have seen that $(D\Lambda)_{ij}=\frac{d_{ij}\alpha_j}{(1+K_jx_j)^c}$.
As a result we have 
$$ \frac{\partial (D\Lambda)_{ij}}{\partial x_l}=-\delta_{jl}\frac{cd_{ij}\alpha_jK_j}{(1+K_jx_j)^x}.$$
which shows that $\frac{\partial (D\Lambda)_{ij}}{\partial x_l}\leq 0$ and that if $j=l$, then $\frac{\partial (D\Lambda)_{ij}}{\partial x_l}<0$.

\textbf{Assumption A6: } Since $$D\Lambda x= \sum_{i=1}^n \frac{d_{1i}\alpha_ix_i}{(1+K_ix_i)^c} $$
we get
\[
\frac{\partial F_i}{\partial x_j}=\frac{d_{ij}\alpha_j+ (1-c)K_j\alpha_jx_j}{(1+K_jx_j)^c}.
\]
Since we have assumed that $c\in (0,1]$ the above implies $\frac{\partial F_i}{\partial x_j}\geq 0$.

Using these computations, we can argue as in the Beverton-Holt case from Appendix \ref{a:BH_d} and prove Theorem \ref{t:stoc_m_h}.

\bibliography{LV}

\end{document}